\documentclass[preprint,12pt]{elsarticle}

\usepackage[margin=3cm]{geometry}

\usepackage{enumitem}
\usepackage{amsmath}
\numberwithin{equation}{section}
\usepackage{amssymb}
\usepackage{amscd}
\usepackage{amsfonts}
\usepackage{amsthm}
\usepackage{xcolor}
\usepackage{soul}
\soulregister\cite7
\soulregister\ref7
\soulregister\eqref7
\usepackage{framed}
\definecolor{newblue}{rgb}{0.80,0.89,1.00}
\definecolor{shadecolor}{rgb}{0.87,0.93,1.00}

\usepackage{graphicx}
\usepackage{hyperref}
\usepackage{tikz}
\usetikzlibrary{arrows.meta, positioning, calc, decorations.pathreplacing}

\newtheorem{theorem}{Theorem}[section]

\newtheorem{proposition}[theorem]{Proposition}
\newtheorem{corollary}[theorem]{Corollary}
\newtheorem{lemma}[theorem]{Lemma}
\theoremstyle{definition}

\newtheorem{remark}[theorem]{Remark}

\newcommand{\beq}{\begin{equation}}
\newcommand{\eeq}{\end{equation}}
\DeclareMathOperator{\arccot}{arccot}
\newcommand{\R}{\mathbb{R}}
\newcommand{\C}{\mathbb{C}}
\newcommand{\dd}{\mathrm{d}}
\newcommand{\wt}{\widetilde}
\newcommand{\RR}{\mathcal R}
\newcommand{\N}{\mathbb{N}}

\newcommand{\ii}{\mathrm{i}}
\newcommand{\Rom}[3]{R_{#1}^{#2,#3}}
\newcommand{\Geg}[2]{\mathcal{C}_{#1}^{#2}}

\newcommand{\Jac}[3]{P_{#1}^{(#2,#3)}}
\newcommand{\ha}{\tfrac12}

\newcommand{\pFq}[2]{{}_{#1}F_{#2}}

\ExplSyntaxOn
\cs_gset:Npn \__first_footerline:
  { \group_begin: \small \sffamily \__short_authors: \group_end: }
\ExplSyntaxOff

\begin{document}
\begin{frontmatter}
\title{Romanovski polynomials, Gegenbauer connections, and $\mathrm{su}(1,1)$ ladder structures}
\author[uned]{J. A. Vallejo\corref{corresponding}}\ead{jvallejo@mat.uned.es}
\address[uned]{Departamento de Matem\'aticas Fundamentales, Universidad Nacional de Educaci\'on a Distancia, Madrid, Spain}
\author[uaslp]{M. Kirchbach}\ead{mariana@ifisica.uaslp.mx}
\address[uaslp]{Instituto de F\'isica, Universidad Aut\'onoma de San Luis Potos\'i, San Luis Potos\'i, Mexico}
\cortext[corresponding]{Corresponding author}
\begin{abstract}
We study the monic Romanovski (pseudo-Jacobi) polynomials with degree dependent
parameters $\beta_K=-K$ and $\alpha_K=c/(K+1)$. We write out the parameter
preserving first order ladder relations at fixed $(\alpha,\beta)$, with real constants.
Combining them returns the three-term recurrence; iterating them builds
$R_n^{\alpha,\beta}$ from the constant polynomial. We solve in triangular form the
connection problem with the $\alpha=0$ family, identified with the Gegenbauer
polynomials, and transfer the ladder relations to the $(n,\ell)$ lattice. Under $x=\cot\chi$ the in-level operators depend on
$\ell$ alone.  The dressed functions carry intrinsic lowest weight $\mathrm{su}(1,1)$
modules along columns of fixed $\ell$; rows need a boundary prescription and a
rescaling. The deformed functions carry the same modules through higher order operators,
which yield relations between Romanovski polynomials with different parameter pairs.
Algebraically the family lies in the Jacobi class, and we indicate throughout which
statements are inherited from that framework by transformation or specialization and
which are not.
\end{abstract}
\end{frontmatter}

\section{Introduction}

The Romanovski polynomials were introduced in 1929 by V.~I.~Romanovski, in the context
of an extension of the Student's $t-$distribution \cite{Romanovski}. A particular case
had already appeared in the work of Routh \cite{Routh}, so the family is also called
Romanovski--Routh or, in Lesky's classification of the finite systems of continuous
classical orthogonal polynomials, pseudo-Jacobi \cite{Lesky,Natanson1}. In the Askey
scheme they occupy the entry \emph{pseudo-Jacobi} \cite[\S9.9]{KLS} (see also
\cite[\S9.9]{Koornwinder}, where the range of admissible parameters is extended). They
are the polynomial solutions of a second order linear differential equation of
hypergeometric type \cite{Nikiforov} on the real line,
\[
(1+x^2)y''(x) + (2\beta x + \alpha)y'(x) - n(2\beta + n - 1)y(x) = 0\,,
\]
and are orthogonal with respect to a weight function that combines a rational factor
$(1+x^2)^{\beta-1}$ with an exponential factor $\exp(-\alpha\arccot x)$. When the
parameters $\alpha$ and $\beta$ are held fixed, only \emph{finitely many} of them are
orthogonal.

However, being of hypergeometric type the family inherits, \emph{at fixed parameters}, the classical three-term recurrence, a Rodrigues formula, and first
order differential lowering and raising relations, which are obtained from the structure
relation $\sigma(x)P_n'=a_nP_{n+1}+b_nP_n+c_nP_{n-1}$ (where $\sigma(x)=1+x^2$) by eliminating one of the two
neighbors \cite{Nikiforov,Koepf,Ismail} (see also \cite{Cotfas,Cotfas2} for the same
construction from the point of view of shape invariance in quantum mechanics). In Sections~\ref{sec:recurrences} and \ref{sec:trig} we give the \emph{explicit real form} 
of the two first order relations in the monic normalization,
with the proportionality constants written out, together with a direct derivation
of the standard three-term recurrence from them. To this
we must add the observation that, after a trigonometric substitution and the choice
$\alpha_K=c/(K+1)$, the resulting in-level\footnote{See Section \ref{sec:basic} below for the terminology regarding levels, multiplets, etc., which is taken from the physical context in which these objects appear.} operators become independent of the level. It
is this explicitness what makes the relations usable in the degree dependent setting we
describe next.

Let us allow the parameters to depend on the degree in the specific manner
\beq\label{eq:degdep}
K = n + \ell\,,\qquad \alpha_K = \frac{c}{K+1}\,,\qquad \beta_K = -K
\eeq
(here $\ell\in\N_0=\N\cup\{0\}$, and $c\in\R$ is a constant). Then, for each fixed $K$,
the finite family $\{\Rom{n}{\alpha_K}{-K}\}_{n=0}^K$ is an orthogonal set with respect
to its own weight. This is the pattern of parameters produced by the trigonometric
Rosen--Morse potential \cite{Compean,Raposo,Weber}, where the label $K$ is the principal
quantum number. The pair $(n,\ell)$ organizes the polynomials on a two-dimensional
lattice, and there appear naturally two different directions on it.

\begin{enumerate}[label=\textit{\alph*})]
\item Along a \emph{column} ($\ell$ fixed, $n$ and hence $K$ varying) both parameters
move with the degree. When $\alpha$ is kept constant this is precisely the family of
complementary polynomials $Q_\nu^{(\alpha,-a)}=R_\nu^{\alpha,-a-\nu}$ studied by Weber
\cite{Weber}, who obtained for it a generating function, a first order recursive
differential relation and a three-term recursion with polynomial coefficients
\cite[Thms.~1.1, 1.2, 1.8]{Weber}. If, moreover, $\alpha$ runs with the degree as in
\eqref{eq:degdep}, those relations do not connect consecutive members of the family any
longer, because $\alpha_{K-1}\neq\alpha_K$. Relations among polynomials with degree
dependent parameters of this kind have been studied, in a rational extension setting, in
\cite{Quesne}. Relations that do connect members with different parameter pairs are
obtained in Section~\ref{subsec:deformed}; their differential order grows with the degree
and their coefficients involve the weight, a behaviour that explains why they can not be achieved with a first order fixed parameter calculus.
\item Along an \emph{anti-diagonal} ($K$ fixed, $n$ and $\ell$ varying in opposite
directions) the pair of parameters $(\alpha_K,\beta_K)$ is constant. Thus, the classical
apparatus applies word by word to each level multiplet, and it does so for the degree
dependent family as well.
\end{enumerate}

This paper is organized around the observation that this 
second item  makes the degree-dependent family tractable. 
To be precise, the relations surviving the degree
dependence are exactly those which do not change the level. Thus, our aim is to develop
them in explicit form (Section~\ref{sec:recurrences}), to show that they generate every
member of a level multiplet from the constant polynomial by first order operations
(Section~\ref{sec:trig}), and to exhibit their trigonometric realization, in which a
single $\ell-$indexed operator acts on all the levels at once. Along the way we record
the identification of the $\alpha=0$ subfamily with the Gegenbauer polynomials and the
resulting Romanovski form of the $S^3$ hyperspherical harmonics
(Section~\ref{sec:relation}), we solve the connection problem between the two parameter
family and its $\alpha=0$ partners (Section~\ref{sec:connection}), and we organize the
four families of first order moves on the lattice, together with the $\mathrm{su}(1,1)$
modules they support (Sections~\ref{sec:ladders} and \ref{sec:visualization}).

\subsection*{Inherited structure and new content}
Following the structural viewpoint of \cite{CastilloGordillo,Castillo}, we separate
explicitly the results inherited from the Jacobi--Gegenbauer framework by transformation,
specialization or conjugation, from the content we regard as new.
The identification with (a subclass of) Jacobi polynomials mentioned above is 
stated and proved in Section~\ref{subsec:jacobi}, where we
also show that at the degree dependent parameters \eqref{eq:degdep} it does not extend
to a Pearson functional on all of $\mathcal P$.
The recurrence at fixed parameters and the first order structure relations are classical
features of the hypergeometric type families \cite{Nikiforov,Koepf,Ismail}, and under
the dictionary \eqref{eq:KLS-dictionary} the recurrence is the one in
\cite[(9.9.4)]{KLS}. The $\alpha=0$ identification and the relation to Jacobi--Gegenbauer
are classical as well \cite{Weber,Wuensche,Szego}. Our contribution here
is their explicit real form in a single monic normalization, with all the constants and
boundary conventions written out, and their organization for the degree dependent
multiplets \eqref{eq:degdep}. The connection coefficients are given directly in the
monic Romanovski parametrization, as a finite subdiagonal sum with a parity property,
and with elementary first subdiagonals. The lattice operators are
conjugates of standard Gegenbauer contiguous relations, and lattices of first order ladder
operators of this kind have been constructed for the Jacobi family, which contains the
ultraspherical one, in \cite{Miller,Wuensche,Celeghini}. There the ambient algebra is
larger than ours, since the Jacobi polynomials carry three independent labels, and
$\mathrm{su}(1,1)$ appears as the subalgebra acting when the remaining labels are held
fixed. It is the collective organization of our operators, together with the
$K-$independent trigonometric in-level realization, what (in our view) make them
interesting.

\section{Preliminaries}
\label{sec:basic}

Let $\alpha,\beta\in\R$. The \emph{Romanovski weight function} is
\begin{equation}\label{eq:weight}
w^{\alpha,\beta}(x) := (1+x^2)^{\beta-1}\exp\,(-\alpha\arccot x)\,,\qquad x\in\R,
\end{equation}
where $\arccot x\in(0,\pi)$ denotes the principal branch. Note that $\arccot x\to0$ as
$x\to+\infty$, while $\arccot x\to\pi$ as $x\to-\infty$, so the exponential factor is
bounded and bounded away from zero on $\R$, and it does not affect the integrability.
Putting $\sigma(x)=1+x^2$ and $\tau(x)=2\beta x+\alpha$, the weight verifies the Pearson
equation $(\sigma w)'=\tau w$.

The \emph{Romanovski polynomials} $\RR_n^{\alpha,\beta}$ of degree $n$ are the polynomial
solutions (unique up to normalization for nondegenerate parameter values, as explained in
Remark~\ref{rem:normalization}) of
\begin{equation}\label{eq:RomanovskiDE}
(1+x^2)y''(x) + (2\beta x + \alpha)y'(x) - n(2\beta + n - 1)y(x) = 0.
\end{equation}

These are real polynomials on the real line. They do not, however, constitute a new
algebraic family, as the substitution $x\mapsto\ii x$ together with the complex conjugate
parameter pair $p=\beta-1+\ii\alpha/2$, $\bar p$, identifies them with Jacobi
polynomials, and in the functional analytic classifications of
\cite{Castillo,CastilloGordillo} they belong to the Jacobi equivalence class, the finite
orthogonality and the existence of a distinguished real weight being features of a
particular realization rather than marks of a distinct algebraic structure (we shall in 
fact exploit their complex Jacobi representation in the proof of Lemma~\ref{lem:anr}). 
Their real realization provides a setting in which the polynomials themselves, the weight, 
and every constant appearing in Sections~\ref{sec:recurrences} and \ref{sec:trig} are 
real; the availability of these explicit real expressions, with all constants
written out, makes the formulas useful in the applications where these polynomials
arise.

Comprehensive modern accounts can be found in \cite{Raposo}, \cite{Natanson1,Natanson2}
and \cite[\S9.9]{KLS}. For their orthogonality and asymptotics outside the classical
parameter range we refer the reader to \cite{Jordaan}, and for a related family on the
unit circle to \cite{MFSRT}. They admit the Rodrigues representation
\begin{equation}\label{eq:RomanovskiRodrigues}
\RR_n^{\alpha,\beta}(x)
=\frac{1}{w^{\alpha,\beta}(x)}\frac{\dd^n}{\dd x^n}
\Bigl(w^{\alpha,\beta}(x)(1+x^2)^n\Bigr)\,,
\end{equation}
once a multiplicative constant has been chosen. Dividing by the Pochhammer symbol
$(2\beta+n-1)_n$ makes the right-hand side monic, and this is the Rodrigues type formula
\cite[(9.9.9)]{KLS}. Throughout the paper the Pochhammer symbol is understood as the
finite product
\begin{equation}\label{eq:pochhammer}
(a)_k:=a(a+1)\cdots(a+k-1)\quad (k\ge1),\qquad (a)_0:=1,
\end{equation}
which is defined for every $a\in\C$, and in particular for a nonpositive integer $a$. It
should \emph{not} be read as the quotient $\Gamma(a+k)/\Gamma(a)$: the two agree whenever
$a$ is not a pole of $\Gamma$, but the factors $(-n)_r$ appearing in
\eqref{eq:Romanovski-monomial-coefficients} below --- which are precisely what truncates
the hypergeometric series there --- are finite and generally nonzero for $r\le n$, while
the corresponding quotient of Gamma functions is undefined.

The monic pseudo-Jacobi polynomials $P_n(x;\nu,N)$ of \cite[(9.9.1)]{KLS} verify the
same equation under
\begin{equation}
\nu=\tfrac{\alpha}{2},\qquad N=-\beta,\qquad
\Rom{n}{\alpha}{\beta}(x)=P_n\bigl(x;\tfrac{\alpha}{2},-\beta\bigr).
\label{eq:KLS-dictionary}
\end{equation}
This same dictionary follows from the weights, because $\arccot x=\pi/2-\arctan x$, and
the resulting factor $e^{-\alpha\pi/2}$ is irrelevant for the orthogonality. We will use
\eqref{eq:KLS-dictionary} below, in order to compare the recurrence and the shift
relations with \cite[\S9.9]{KLS}.
\begin{remark}\label{rem:normalization}
Throughout the paper we use the monic normalization, and we distinguish the monic
Romanovski polynomials by a normal, non calligraphic font. Thus,
$\Rom{n}{\alpha}{\beta}(x)= x^n+\text{lower-degree terms}$ denotes the unique monic
polynomial solution of \eqref{eq:RomanovskiDE}, for those parameter values such that the
denominators $(k-n)(k+n+2\beta-1)$, $0\le k<n$, of the coefficient recurrence
\eqref{eq:CoeffRecursionSolved} below do not vanish, that is,
\begin{equation}\label{eq:nondeg}
2\beta+m+k-1\neq0\qquad\text{for }0\le k<m,
\end{equation}
with $m=n$. Otherwise both existence and uniqueness can fail. For $\beta=-1$ and
$n=2$ the condition \eqref{eq:nondeg} is violated at $k=1$; substituting
$y=x^{2}+bx-1$ into \eqref{eq:RomanovskiDE} leaves the residual $\alpha(2x+b)$, so that
for $\alpha=0$ every monic polynomial $x^{2}+bx-1$, $b\in\R$, is a solution and
uniqueness fails, whereas for $\alpha\neq0$ there is no monic solution of degree $2$ at
all. This is the general pattern at $\beta=-K$, taking $m=K+1$ and $k=K$ in the
coefficient recurrence \eqref{eq:CoeffRecurrence} below, where $r_{m,m}=1$ and
$r_{m,m+1}=0$, the coefficient of $r_{m,K}$ vanishes and what remains is
$\alpha(K+1)=0$; so for $\alpha\neq0$ no monic solution of degree $K+1$ exists, while
for $\alpha=0$ the coefficient $r_{K+1,K}$ is free and uniqueness fails. We will refer to \eqref{eq:nondeg} as the \emph{nondegeneracy
condition at degree $m$}. All the parameter values used in this paper verify it. In
particular $\beta=-K$ with $0\le m\le K$ does, because in that case $k\le m-1$ gives
$m+k\le 2m-1\le 2K-1$, and hence $2\beta+m+k-1=-2K+m+k-1\le-2$. The calligraphic symbol
$\RR_n^{\alpha,\beta}$ is reserved for the Rodrigues normalization
\eqref{eq:RomanovskiRodrigues}.
\end{remark}

Let us substitute $\Rom{n}{\alpha}{\beta}(x)=\sum_{j=0}^n r_{n,j}x^j$ with $r_{n,n}=1$
into \eqref{eq:RomanovskiDE}. The coefficient of $x^k$ in $(1+x^2)y''$ is
$(k+2)(k+1)r_{n,k+2}+k(k-1)r_{n,k}$ (here $r_{n,j}:=0$ for $j\notin\{0,1,\dots,n\}$, a
convention which will be used again in the proof of Theorem~\ref{inlevel:thm:lower}),
the one in $(2\beta x+\alpha)y'$ is $2\beta k r_{n,k}+\alpha(k+1)r_{n,k+1}$, and the one
in $-n(2\beta+n-1)y$ is $-n(2\beta+n-1)r_{n,k}$. Now, since $k(k-1)+2\beta
k-n(2\beta+n-1)=(k-n)(k+n+2\beta-1)$, we arrive at the fundamental coefficient
recurrence
\begin{equation}
(k+2)(k+1)r_{n,k+2}+\alpha(k+1)r_{n,k+1}+(k-n)(k+n+2\beta-1)r_{n,k}=0\,,
\label{eq:CoeffRecurrence}
\end{equation}
which, starting from $r_{n,n}=1$, $r_{n,n+1}=0$, determines all the lower coefficients
by descending induction,
\begin{equation}
r_{n,k}=-\frac{(k+2)(k+1)r_{n,k+2}+\alpha(k+1)r_{n,k+1}}{(k-n)(k+n+2\beta-1)}\,,
\label{eq:CoeffRecursionSolved}
\end{equation}
whenever the denominator is nonzero. In particular,
\begin{equation}\label{eq:subleading-coeff}
r_{n,n-1}=\frac{n\alpha}{2(\beta+n-1)}\,,\qquad
r_{n,n-2}=\frac{n(n-1)\bigl[2(\beta+n-1)+\alpha^{2}\bigr]}
{4(\beta+n-1)(2\beta+2n-3)}\,.
\end{equation}

Whenever rational expressions appear, one must be cautious with the denominators, so
let us fix the cancellation convention we will adhere to in what follows.

\begin{remark}
\label{rem:cancellation}
Whenever a displayed rational coefficient has a factor in the numerator and a factor in
the denominator which vanish simultaneously, its value is understood after cancellation
or, equivalently, by rational continuation through the removable
singularity. True poles are always excluded by hypothesis. This applies in particular to
\eqref{eq:subleading-coeff}, to the recurrence coefficients \eqref{eq:3TR-coeffs}, to
the ladder constant \eqref{eq:inlevel-low-const} and to the closed form
\eqref{eq:Romanovski-monomial-coefficients}. In every one of these cases the continued
value is the one produced by \eqref{eq:CoeffRecursionSolved}, so no ambiguity arises at
all.
\end{remark}

\subsection{Finite orthogonality and level multiplets}
\label{subsec:finite}
At fixed parameters the weight \eqref{eq:weight} decays only polynomially, so that only
a finite number of its moments exist, and only finitely many Romanovski polynomials are
orthogonal with respect to it \cite{Lesky,Jamei,Raposo}. However, for the choice
\eqref{eq:degdep} the polynomials
$\{\Rom{n}{\alpha_K}{-K}\}_{n=0}^{K}$ are mutually orthogonal with respect to
$w^{\alpha_K,-K}$. Indeed, the inner products converge because their integrands are
$O(|x|^{n+m-2K-2})$, and the self adjoint form $(\sigma w y')'=n(2\beta+n-1)wy$ gives,
for $m<n\le K$,
\[
[\sigma w(R_n'R_m-R_nR_m')]_{-\infty}^{\infty}
=(\varrho_n-\varrho_m)\int_{\mathbb R}R_nR_mw\,\dd x.
\]
Here the boundary term is zero, while $\varrho_n-\varrho_m=(n-m)(n+m-2K-1)\ne0$. 

With the lattice terminology introduced in the Introduction, a \emph{level} is the anti-diagonal on which $ K=n+\ell $ is fixed. Since both $\alpha_K$ and $\beta_K$ depend only on $K$, all the polynomials lying on the same level share the parameter pair $(\alpha_K,\beta_K)$. The corresponding \emph{level multiplet} is \[ \mathcal{M}_K := \left\{ \Rom{n}{\alpha_K}{\beta_K} :\ 0\le n\le K,\ \ell=K-n \right\}. \] We call a relation or an operator \emph{in-level} when it preserves $K$, and hence maps members of $\mathcal{M}_K$ to members of the same multiplet. Equivalently, an in-level move changes $n$ and $\ell$ by opposite amounts, so that $ \Delta K=\Delta n+\Delta\ell=0$. In particular, it preserves the pair $(\alpha_K,\beta_K)$, which is the essential reason why the fixed-parameter differential relations remain applicable to the degree dependent family.

Two comments about the particular choice \eqref{eq:degdep} are in order. 
First, the dependence of $\beta_K$ on
the degree is linear, while the one of $\alpha_K$ is rational, and it is this last one
what obstructs the column relations recalled in the Introduction. Second, the precise
form $\alpha_K=c/(K+1)$ is dictated by the trigonometric Rosen--Morse problem
\cite{Compean,Raposo}, but it plays no role in
Sections~\ref{sec:relation}--\ref{sec:recurrences}. All the statements there hold for an
arbitrary $\alpha_K$, because $\alpha_K$ is simply a constant once $K$ has been fixed.
It is used for the first time in Section~\ref{sec:trig}, where it is exactly what makes
the trigonometric in-level operators independent of $K$.

\subsection{The three-term recurrence}

Since $\alpha_K$ and $\beta_K$ are constant inside a level, the level multiplet is a
finite orthogonal polynomial sequence in the usual sense and, in particular, it verifies
the standard three-term recurrence. By \eqref{eq:KLS-dictionary} this is exactly the
normalized recurrence relation \cite[(9.9.4)]{KLS} for the pseudo-Jacobi polynomials,
and, as these are already monic, no conversion of normalization is involved. Indeed,
substituting $\nu=\alpha/2$, $N=-\beta$ into \cite[(9.9.4)]{KLS} and using
$(n-N-1-i\nu)(n-N-1+i\nu)=(\beta+n-1)^{2}+\tfrac14\alpha^{2}$, one reproduces
\eqref{eq:3TR-coeffs} below. Equivalently, by Proposition~\ref{prop:jacobi-identity}
these are the monic Jacobi recurrence coefficients $B^{J}_n$, $C^{J}_n$ read through
$z=\ii x$, namely $b_n=B^{J}_n/\ii$ and $c_n=-C^{J}_n$. We restate the recurrence in the
$(\alpha,\beta)$ variables because the coefficients in this form are the ones we will
use in Sections~\ref{sec:recurrences} and \ref{sec:trig}. The proof is deferred to
Corollary~\ref{cor:3TR}.

\begin{proposition}\label{prop:3TR-statement}
Adopt the convention $\Rom{-1}{\alpha}{\beta}\equiv0$. Let $n\ge0$, and let
$\alpha,\beta\in\R$ verify the nondegeneracy condition \eqref{eq:nondeg} at the degrees
$n$ and $n+1$, and also at degree $n-1$ when $n\ge1$. Then, in the monic normalization,
\begin{equation}\label{eq:3TR}
x\,\Rom{n}{\alpha}{\beta}(x)=\Rom{n+1}{\alpha}{\beta}(x)+b_n\Rom{n}{\alpha}{\beta}(x)
+c_n\Rom{n-1}{\alpha}{\beta}(x),
\end{equation}
with
\begin{equation}\label{eq:3TR-coeffs}
\begin{split}
b_n=& \frac{\alpha(1-\beta)}{2(\beta+n-1)(\beta+n)}\,,\\
c_n=& -\,\frac{n\,(2\beta+n-2)\bigl[\alpha^{2}+4(\beta+n-1)^{2}\bigr]}
{4(\beta+n-1)^{2}(2\beta+2n-3)(2\beta+2n-1)}\,.
\end{split}
\end{equation}
\end{proposition}

At $n=0$ the last term is absent, both because $c_0=0$ and because
$\Rom{-1}{\alpha}{\beta}=0$, and \eqref{eq:3TR} reduces to
$x=\Rom{1}{\alpha}{\beta}(x)+b_0$, with $b_0=-\alpha/(2\beta)$. By
Remark~\ref{rem:cancellation}, at $n=0$ the expressions \eqref{eq:3TR-coeffs} give
$b_0=-\alpha/(2\beta)$ and $c_0=0$, the unsimplified form of $b_0$ being indeterminate
at $\beta=1$, although the cancelled value is not. Also, at $n=1$ the factor
$2\beta+2n-3=2\beta-1$ cancels between the numerator and the denominator of $c_n$, and
what remains is $c_1=-(\alpha^{2}+4\beta^{2})/\bigl(4\beta^{2}(2\beta+1)\bigr)$. For the
level multiplet, with $\beta=\beta_K=-K$, $\ell=K-n\ge1$ and $c=(K+1)\alpha_K$, the
coefficients become
\begin{equation}\label{eq:3TR-level}
b_n=\frac{c}{2\,\ell(\ell+1)}\,,
\qquad
c_n=\frac{n(n+2\ell+2)}{(2\ell+1)(2\ell+3)}
\left[1+\frac{\alpha_K^{2}}{4(\ell+1)^{2}}\right]>0 .
\end{equation}

The positivity of $c_n$ in \eqref{eq:3TR-level} is what one expects for
a positive definite functional. Note also that $b_n$ depends there on $\ell$ alone. At
$\ell=0$, that is, at $n=K$, the expression for $b_n$ is singular, while $c_n$ has the
finite limit $\tfrac13 n(n+2)\bigl[1+\tfrac14\alpha_K^{2}\bigr]$. What fails at the top
of a level is not the backward coefficient, but the forward term. Indeed, with
$\beta=-K$ and $m=K+1$, $k=K$, the left-hand side of \eqref{eq:nondeg} vanishes, so
$\Rom{K+1}{\alpha_K}{-K}$ is degenerate and it is not defined by
Remark~\ref{rem:normalization}; as observed there, for $\alpha_K\neq0$ no monic
solution of that degree exists at all. Thus, the recurrence \eqref{eq:3TR} runs only for $0\le
n\le K-1$, which is as far as a finite orthogonal family allows. The value
of $c_n$ at $n=K$ still has a meaning in the lowering relation
\eqref{eq:inlevel-low-phys}, but not as a part of a forward three-term recurrence. This
is the same boundary effect which will appear again as the restriction $\ell\ge1$ in all
the raising relations below.

\subsection{Gegenbauer conventions}

For $\lambda>-\ha$ the Gegenbauer (ultraspherical) polynomials $\Geg{n}{\lambda}$ are
defined by the generating function $\sum_{n\ge0}\Geg{n}{\lambda}(u)t^n=(1-2ut+t^2)^{-\lambda}$;
and for $\lambda>-\ha$, $\lambda\ne0$, they are orthogonal on $[-1,1]$ with respect to
$(1-u^2)^{\lambda-1/2}$. We use the normalization of Szeg\H{o}~\cite{Szego}, in which
\begin{equation}
\Geg{0}{\lambda}(u)=1,\qquad \Geg{1}{\lambda}(u)=2\lambda u,\qquad
\Geg{n}{\lambda}(u)=\frac{2^{n}(\lambda)_n}{n!}\,u^n+\dots ,
\label{eq:gegenbauer-leading}
\end{equation}
with $(a)_n$ as in \eqref{eq:pochhammer}, and $\Geg{n}{\lambda}(1)=(2\lambda)_n/n!$. Throughout
the paper we adopt the conventions $\Geg{-1}{\lambda}\equiv0$ and, at $\lambda=0$,
$\Geg{0}{0}\equiv1$ while $\Geg{n}{0}\equiv0$ for $n\ge1$. This last one is a
degeneration of the chosen normalization, and not of the underlying polynomial system,
because $\lim_{\lambda\to0}\lambda^{-1}\Geg{n}{\lambda}=(2/n)\,T_{n}$ is finite and
nonzero for $n\ge1$ \cite[Eq.~(4.7.8)]{Szego}. The Gegenbauer polynomials are related to
the Jacobi polynomials with equal parameters by
\begin{equation}
\Jac{n}{\delta}{\delta}(u)
=\frac{\Gamma(n+\delta+1)\,\Gamma(2\delta+1)}{\Gamma(\delta+1)\,\Gamma(n+2\delta+1)}\,
\Geg{n}{\delta+1/2}(u),
\label{eq:gegenbauer-jacobi} \end{equation} and they verify
\begin{equation}
\Bigl[(1-u^2)\tfrac{\dd^2}{\dd u^2}-(2\lambda+1)u\tfrac{\dd}{\dd u}
+n(n+2\lambda)\Bigr]\Geg{n}{\lambda}(u)=0 .
\label{eq:gegenbauer-ode}
\end{equation}

\subsection[Relation to the Jacobi class]{Relation to the Jacobi class}
\label{subsec:jacobi}
\noindent
We state the identification with the Jacobi family in the monic normalization used
throughout. Let $\Jac{n}{p}{q}$ denote the Jacobi polynomial in
its standard normalization, let $J_n^{(p,q)}$ denote its monic version, and put
\begin{equation}
p=\beta-1+\tfrac{\ii\alpha}{2}\,,\qquad q=\bar p=\beta-1-\tfrac{\ii\alpha}{2}\,,
\label{eq:pq}
\end{equation}
so that $p+q+2=2\beta$ and $q-p=-\ii\alpha$ are real.

\begin{proposition}\label{prop:jacobi-identity}
Let $\alpha,\beta\in\R$ verify the nondegeneracy condition \eqref{eq:nondeg} at degree
$n$. Then $J_n^{(p,q)}$ is defined and
\begin{equation}
\Rom{n}{\alpha}{\beta}(x)=\ii^{-n}J_n^{(p,q)}(\ii x)\,.
\label{eq:jacobi-identity}
\end{equation}
\end{proposition}
\begin{proof}
The coefficient of $z^{n}$ in $\Jac{n}{p}{q}$ is $(n+p+q+1)_n/(2^{n}n!)=
(n+2\beta-1)_n/(2^{n}n!)$, whose factors are $2\beta+n+k-1$ with $0\le k<n$ by
\eqref{eq:pochhammer}; so it is nonzero precisely under \eqref{eq:nondeg} at degree $n$,
and $J_n^{(p,q)}$ is then well defined. The Jacobi equation
$(1-z^{2})y''+[q-p-(p+q+2)z]y'+n(n+p+q+1)y=0$ reads here
$(1-z^{2})y''-(\ii\alpha+2\beta z)y'+n(n+2\beta-1)y=0$, and the substitution $z=\ii x$,
under which $1-z^{2}=1+x^{2}$ and $\dd/\dd z=-\ii\,\dd/\dd x$, turns it into
\eqref{eq:RomanovskiDE}. Thus $\ii^{-n}J_n^{(p,q)}(\ii x)$ is a monic polynomial
solution of \eqref{eq:RomanovskiDE} of degree $n$, and by
Remark~\ref{rem:normalization} the monic solution is unique under \eqref{eq:nondeg}.
\end{proof}

Equivalently, with $(Tg)(x):=g(\ii x)$ the same substitution conjugates the Jacobi
operator into minus the operator of \eqref{eq:RomanovskiDE}. At the parameters
\eqref{eq:degdep} it gives $\Rom{n}{\alpha_K}{-K}(x)=\ii^{-n}
J_n^{(-K-1+\ii\alpha_K/2,\,-K-1-\ii\alpha_K/2)}(\ii x)$ for $0\le n\le K$, so that a
level is a finite segment of one fixed parameter Jacobi family, while varying $K$
changes the parameters. Consequently the three-term recurrence \eqref{eq:3TR}, the first
order structure relations of Section~\ref{sec:recurrences} and the connection
coefficients of Section~\ref{sec:connection} are all inherited, and we point this out at
each place where they occur; in the classifications of \cite{Castillo,CastilloGordillo}
the family lies in the Jacobi equivalence class. The lattice of labels $(n,\ell)$ used
here records degrees and orders, and is unrelated to a lattice for the independent
variable in a divided difference calculus.

Algebraic equivalence does not amount to the existence of a regular functional on all of
$\mathcal P$. At the parameters \eqref{eq:degdep} no such functional exists, as we now
check. Fix $K\ge1$, write $a=\alpha_K$, and set
\[
\mu_j:=\int_{\R}x^{j}w^{a,-K}(x)\,\dd x ,
\]
which converges absolutely precisely for $0\le j\le2K$, the integrand being
$O(|x|^{\,j-2K-2})$. Multiplying the Pearson equation
$\bigl((1+x^{2})w^{a,-K}\bigr)'=(-2Kx+a)w^{a,-K}$ by $x^{j}$ and integrating by parts we
get
\begin{equation}
(j-2K)\mu_{j+1}+a\,\mu_j+j\,\mu_{j-1}=0\,,\qquad 0\le j\le2K-1,
\label{eq:moment-system}
\end{equation}
the last term being absent at $j=0$, and the boundary term
$\bigl[x^{j}(1+x^{2})^{-K}e^{-a\arccot x}\bigr]_{-\infty}^{+\infty}$ vanishing there
because $j\le2K-1$. At the next index it does not. Indeed, the leading powers cancel in
the derivative,
\[
\frac{\dd}{\dd x}\Bigl[x^{2K}(1+x^{2})w^{a,-K}(x)\Bigr]
=2Kx^{2K-1}w^{a,-K}(x)+a\,x^{2K}w^{a,-K}(x)\,,
\]
whose right-hand side is absolutely integrable, so that integrating over $\R$ and using
$\arccot x\to0$ as $x\to+\infty$ and $\arccot x\to\pi$ as $x\to-\infty$,
\begin{equation}
a\,\mu_{2K}+2K\mu_{2K-1}
=\Bigl[x^{2K}(1+x^{2})^{-K}e^{-a\arccot x}\Bigr]_{-\infty}^{+\infty}
=1-e^{-\pi a}\,.
\label{eq:moment-obstruction}
\end{equation}
Now, for a linear functional $u$ on $\mathcal P$ the Pearson equation
$\mathrm D\bigl((1+x^{2})u\bigr)=(-2Kx+a)u$, with
$\langle\mathrm Du,f\rangle:=-\langle u,f'\rangle$, is equivalent, after testing on
monomials, to \eqref{eq:moment-system} holding for \emph{every} $j\ge0$. Since
\eqref{eq:moment-system} determines $\mu_1,\dots,\mu_{2K}$ from $\mu_0$, while the
coefficient of $\mu_{2K+1}$ vanishes at $j=2K$, for $a\neq0$ no such functional with
nonzero zeroth moment exists; its moments would be proportional to the $\mu_j$, and
\eqref{eq:moment-obstruction} would force a nonzero multiple of $1-e^{-\pi a}$ to
vanish. For $a=0$ the compatibility condition does hold, the moment of order $2K+1$ is
free, and the equation at $j=2K+1$ gives
\[
\mu_{2K+2}=-(2K+1)\mu_{2K}<0 ,
\]
so that no extension is positive definite on $\R$. What remains valid in both cases is
the finite orthogonality of Section~\ref{subsec:finite}: the weight defines a functional
on $\mathcal P_{2K}$, positive definite on $\mathcal P_{K}$, and not an integration over
all of $\mathcal P$. Proposition~\ref{prop:jacobi-identity} is thus an identity between
polynomials, but it does not transport the weight to a regular Pearson
functional on $\mathcal P$.

\section{The $\alpha=0$ subfamily and the Gegenbauer connection}
\label{sec:relation}
For $\alpha=0$, the defining equation \eqref{eq:RomanovskiDE} reads
\begin{equation}\label{RomDE0}
(1+x^2)y''+2\beta x y'-n(2\beta+n-1)y=0.
\end{equation}
Comparing the differential equations and the leading coefficients gives the following
identity, which is the Jacobi--Gegenbauer transformation \eqref{eq:gegenbauer-jacobi}
written for equal Jacobi parameters.
\begin{theorem}\label{thm:master}
For $n,\ell\in\N_0$, $K=n+\ell$, and $\beta_K=-K$,
\begin{align}
&\Rom{n}{0}{-K}(x)
 =c(n,\ell)(1+x^2)^{n/2}
   \Geg{n}{\ell+1}\!\left(\frac{x}{\sqrt{1+x^2}}\right),
\label{eq:master-x}
\end{align}
or, in its dressed version after a trigonometric change of variable,
\begin{align}   
&\sin^n\chi\,\Rom{n}{0}{-K}(\cot\chi)
 =c(n,\ell)\Geg{n}{\ell+1}(\cos\chi),
\label{eq:master}
\end{align}
where
\[
c(n,\ell)=\frac{n!(2\ell+1)!}{(n+2\ell+1)!}
=\binom{n+2\ell+1}{n}^{-1}.
\]
Consequently, we have the identity for dressed Romanovski polynomials
\begin{equation}\label{eq:master-K}
\sin^K\chi\,\Rom{n}{0}{-K}(\cot\chi)
=c(n,\ell)\sin^\ell\chi\,\Geg{n}{\ell+1}(\cos\chi).
\end{equation}
\end{theorem}
\begin{proof}
Let $s:=\sqrt{1+x^2}$ and $u:=x/s$, so that $u'=s^{-3}$ and $1-u^{2}=s^{-2}$, and put
$g:=c(n,\ell)\,s^{n}\Geg{n}{\ell+1}(u)$.

First, $g$ is a monic polynomial of degree $n$. Indeed, by the parity
$\Geg{n}{\lambda}(-u)=(-1)^n\Geg{n}{\lambda}(u)$ only the monomials $u^k$ with $k\equiv
n\ (\mathrm{mod}\ 2)$ occur, and each one of them contributes $x^{k}s^{\,n-k}$ with $n-k$
even, hence a polynomial; and since $u\to1$ as $x\to+\infty$, the coefficient of $x^n$ is
\[
c(n,\ell)\,\Geg{n}{\ell+1}(1)=\frac{n!\,(2\ell+1)!}{(n+2\ell+1)!}\cdot\frac{(2\ell+2)_n}{n!}=1 ,
\]
because $(2\ell+2)_n=(n+2\ell+1)!/(2\ell+1)!$.

Second, $g$ solves \eqref{RomDE0} with $\beta=\beta_K=-K$. Writing $y=\Geg{n}{\ell+1}$
one has $g'=c(n,\ell)s^{\,n-3}\bigl(n\,x\,s\,y+y'\bigr)$, and, after eliminating $y''$ by
means of the Gegenbauer equation \eqref{eq:gegenbauer-ode} at $\lambda=\ell+1$, a direct
computation gives
\[
(1+x^{2})g''-2(n+\ell)x\,g'+n(n+2\ell+1)\,g=0.
\]
Since $\beta_K=-K=-(n+\ell)$ and $-n(2\beta_K+n-1)=n(n+2\ell+1)$, this is precisely
\eqref{RomDE0}.

By Remark~\ref{rem:normalization} the monic solution of \eqref{RomDE0} of degree $n$ is
unique for $\beta_K=-K$, so $g=\Rom{n}{0}{-K}$, and this is \eqref{eq:master-x}. Putting
now $x=\cot\chi$, so that $s^{-1}=\sin\chi$ and $u=\cos\chi$ on $(0,\pi)$, and
multiplying by $\sin^n\chi$, we get \eqref{eq:master}; a further factor $\sin^{\ell}\chi$
gives \eqref{eq:master-K}.
\end{proof}

Thus, the dressed $\alpha=0$ Romanovski polynomials give the usual unnormalized $S^3$
hyperspherical harmonics. In particular,
\begin{equation}\label{eq:harmonics}
Y_{K\ell m}=c(n,\ell)^{-1}\sin^K\chi\,
\Rom{n}{0}{-K}(\cot\chi)Y_\ell^m(\theta,\varphi).
\end{equation}
In particular, they are eigenfunctions of $-\Delta_{S^3}+1$ with eigenvalue
$K(K+2)+1=(K+1)^2$, where $\Delta_{S^3}$ denotes the Laplace--Beltrami operator of the
round unit sphere. Of course, the numerical coincidence of this eigenvalue with the
dimension $(K+1)^2$ of the eigenspace of degree $K$ should not be confused with an
identity of invariants.

\section{The connection problem at fixed $\beta$}\label{sec:connection}

Now we are going to expand a general Romanovski polynomial in the basis provided by its
$\alpha=0$ partners with the same second parameter. Let us fix $K\in\N_0$ and
$\beta_K=-K$. The polynomials $\Rom{m}{0}{\beta_K}$, $0\le m\le K$, all share the weight
$w^{0,-K}$ and form an orthogonal system in $L^2(\R,(1+x^2)^{-K-1}\dd x)$, so any
polynomial of degree $n\le K$ has a unique expansion in this finite basis. As
$\Rom{n}{\alpha_K}{-K}$ has degree $n$, we can write
\begin{equation}\label{conn1}
\Rom{n}{\alpha_K}{-K}(x)=\sum_{m=0}^n c_{nm}^K\,\Rom{m}{0}{-K}(x),
\qquad
c_{nm}^K=\frac{\bigl\langle \Rom{n}{\alpha_K}{-K},\Rom{m}{0}{-K}\bigr\rangle_{w^{0,-K}}}
{\bigl\langle \Rom{m}{0}{-K},\Rom{m}{0}{-K}\bigr\rangle_{w^{0,-K}}}.
\end{equation}

\begin{proposition}\label{prop:connection}
For each fixed level $K$ the $(K+1)\times(K+1)$ matrix $(c_{nm}^K)_{0\le n,m\le K}$ is
unit lower triangular, that is, $c_{nm}^K=0$ for $m>n$ and $c_{nn}^K=1$.
\end{proposition}

\begin{proof}
The $\Rom{m}{0}{-K}$ are monic of degrees $0,\dots,K$, so they form a graded basis of
$\mathcal P_K$. Therefore the expansion of $\Rom{n}{\alpha_K}{-K}\in\mathcal P_n$
involves only $m\le n$, and comparing the leading coefficients of the two monic sides
gives $c_{nn}^K=1$.
\end{proof}

The connection coefficients between Jacobi polynomials with different parameters are
classical \cite{Askey,Rainville}. By Proposition~\ref{prop:jacobi-identity} the
coefficients below are exactly the monic Jacobi ones: at $\alpha=0$ that proposition
gives $\Rom{m}{0}{\beta}(x)=\ii^{-m}J_m^{(\beta-1,\beta-1)}(\ii x)$, so that if
$J_n^{(p,q)}=\sum_{m=0}^{n}d_{nm}J_m^{(\beta-1,\beta-1)}$ on a nondegenerate segment,
then $c_{nm}(\alpha,\beta)=\ii^{\,m-n}d_{nm}$. This is an identity between finite
expansions, with no question of convergence.
We prefer the following derivation, which produces a form adapted to the monic
normalization. The first step consists in inverting the monomial basis.

\begin{lemma}\label{lem:monomial-inversion}
Let $\beta\in\R$ be such that the monic polynomials $\Rom{m}{0}{\beta}$, $0\le m\le r$,
are nondegenerate in the sense of Remark~\ref{rem:normalization} and such that
$\bigl(\beta+r-2s+\ha\bigr)_s\neq0$ for $0\le s\le\lfloor r/2\rfloor$. Then, for
every $r\ge0$,
\begin{equation}
x^r=\sum_{s=0}^{\lfloor r/2\rfloor}\binom{r}{2s}
\frac{(-1)^s\left(\ha\right)_s}{\left(\beta+r-2s+\ha\right)_s}\,\Rom{r-2s}{0}{\beta}(x)\,.
\label{eq:monomial-inversion}
\end{equation}
\end{lemma}
\begin{proof}
Let us expand $x^r$ in the orthogonal monic basis and use the monic form of
\eqref{eq:RomanovskiRodrigues} at $\alpha=0$, that is,
\[
\Rom{m}{0}{\beta}(x)=\frac{(1+x^{2})^{1-\beta}}{(2\beta+m-1)_m}
\frac{\dd^{m}}{\dd x^{m}}(1+x^{2})^{\beta+m-1}.
\]
Assume first that $\beta<\tfrac{1-2r}{2}$, so that all the integrals converge absolutely
and all the boundary terms vanish. After $m$ integrations by parts the parity forces
$r-m=2s$, and the remaining integral is the beta integral
$\int_{\R}x^{2s}(1+x^{2})^{\beta+m-1}\dd x=B\bigl(s+\ha,\ha-\beta-m-s\bigr)$. Dividing by
the squared norm, which is the same integral at $s=0$, the Rodrigues constant cancels and
one gets the displayed coefficient. Finally, for a fixed $r$ both sides are polynomials in
$x$ whose coefficients are rational functions of $\beta$, so the identity extends from
that range by rational continuation, with Remark~\ref{rem:cancellation} fixing the
removable singularities.
\end{proof}

\begin{lemma}\label{lem:anr}
Set $A:=\beta+\tfrac{\ii\alpha}{2}$ and $N:=n+2\beta-1$, and write
$\Rom{n}{\alpha}{\beta}(x)=\sum_{r=0}^{n}a_{nr}^{(\alpha,\beta)}x^r$, so that
$a_{nr}^{(\alpha,\beta)}=r_{n,r}$ in the notation of \eqref{eq:CoeffRecurrence}. Then
\begin{equation}
a_{nr}^{(\alpha,\beta)}=\ii^{\,r-n}(-1)^r2^{\,n-r}
\frac{(A)_n(-n)_r(N)_r}{r!\,(A)_r(N)_n}\;
\pFq{2}{1}\!\left(\begin{matrix} r-n,\;N+r\\ A+r\end{matrix};\ \ha\right).
\label{eq:Romanovski-monomial-coefficients}
\end{equation}
\end{lemma}
\begin{proof}
Insert the hypergeometric representation of
$P_n^{(\beta-1+\ii\alpha/2,\,\beta-1-\ii\alpha/2)}(\ii x)$, normalize its leading
coefficient to one, expand $(1-\ii x)^k$ and reverse the two finite sums. 
The Pochhammer identity $(a)_{j+r}=(a)_r(a+r)_j$ gives the terminating ${}_2F_1$. 
The property of being real follows independently from the real coefficient recurrence
\eqref{eq:CoeffRecurrence}; this is a concrete instance of the interplay described in
Section~\ref{sec:basic}, where the complex Jacobi representation is used as a
computational device while the final formulas are expressed in the real realization.
The exceptional parameter values are read as in Remark~\ref{rem:cancellation}.
\end{proof}
Here all the Pochhammer symbols are the finite products \eqref{eq:pochhammer}, so that
$(A)_n$, $(A)_r$ and $(A+r)_j$ make sense although $A$ is complex, and $(-n)_r$ makes
sense although $-n$ is a nonpositive integer. As for the denominators, the factors of
$(N)_n$ are $2\beta+n+k-1$ with $0\le k<n$, so that $(N)_n\neq0$ is precisely the
nondegeneracy condition \eqref{eq:nondeg} at degree $n$; and when $\alpha\neq0$ every
factor of $(A)_r$ and of $(A+r)_j$ has imaginary part $\alpha/2\neq0$, hence it is nonzero
as well. Consequently the only apparent singularities of
\eqref{eq:Romanovski-monomial-coefficients} at nondegenerate parameters occur when
$\alpha=0$ and $\beta\in\{0,-1,\dots,-(n-1)\}$, and there
$a_{nr}^{(\alpha,\beta)}$ is \emph{defined} as the coefficient produced by
\eqref{eq:CoeffRecursionSolved}, that is, the value obtained after cancellation and
continuation, in such a way that the coefficient is well defined in the whole
nondegenerate range. Note that this never happens in the level case $\beta=-K$ with
$0\le n\le K$, where \eqref{eq:Romanovski-monomial-coefficients} may be used as it stands.

\begin{theorem}\label{thm:connection-coefficients}
Let $\alpha,\beta\in\R$ verify the hypotheses of Lemma~\ref{lem:monomial-inversion} for
$r\le n$, and define $c_{nm}(\alpha,\beta)$ by
$\Rom{n}{\alpha}{\beta}=\sum_{m=0}^{n}c_{nm}(\alpha,\beta)\Rom{m}{0}{\beta}$. Then
\begin{equation}
c_{nm}(\alpha,\beta)=\sum_{s=0}^{\lfloor(n-m)/2\rfloor}
a_{n,m+2s}^{(\alpha,\beta)}\binom{m+2s}{m}
\frac{(-1)^s\left(\ha\right)_s}{\left(\beta+m+\ha\right)_s}\,,
\label{eq:connection-coefficients-closed}
\end{equation}
with $a_{n,m+2s}^{(\alpha,\beta)}$ given by \eqref{eq:Romanovski-monomial-coefficients}.
Moreover, $c_{nm}(\alpha,\beta)$ is a polynomial in $\alpha$ of degree at most\footnote{In fact of degree exactly $n-m$, with
$[\alpha^{\,n-m}]c_{nm}(\alpha,\beta)=\binom{n}{m}/(n+m+2\beta-1)_{n-m}\neq0$.} $n-m$
and it has the same parity as $n-m$,
\begin{equation}
c_{nm}(-\alpha,\beta)=(-1)^{n-m}c_{nm}(\alpha,\beta).
\label{eq:connection-coefficients-parity}
\end{equation}
\end{theorem}
\begin{proof}
Insert $\Rom{n}{\alpha}{\beta}=\sum_r a_{nr}^{(\alpha,\beta)}x^r$ and apply
Lemma~\ref{lem:monomial-inversion} to each monomial. Collecting the coefficient of
$\Rom{m}{0}{\beta}$ amounts to setting $r=m+2s$, and $\binom{m+2s}{2s}=\binom{m+2s}{m}$.
As for the parity statement, \eqref{eq:RomanovskiDE} is invariant under $x\mapsto-x$,
$\alpha\mapsto-\alpha$, so by uniqueness
$\Rom{n}{-\alpha}{\beta}(x)=(-1)^n\Rom{n}{\alpha}{\beta}(-x)$; since
$\Rom{m}{0}{\beta}(-x)=(-1)^m\Rom{m}{0}{\beta}(x)$, comparison gives
\eqref{eq:connection-coefficients-parity}. Finally, \eqref{eq:CoeffRecurrence} shows
that $a_{nr}^{(\alpha,\beta)}$ is a polynomial in $\alpha$ of degree at most $n-r$, and
in each summand $r=m+2s$, in such a way that the degree is at most $n-m$.
As for the exact degree, in \eqref{eq:CoeffRecursionSolved} the term $r_{n,k+2}$
has $\alpha-$degree at most $n-k-2$, so that it does not reach $\alpha^{\,n-k}$, and
therefore $[\alpha^{\,n-k}]r_{n,k}
=\frac{-(k+1)}{(k-n)(k+n+2\beta-1)}\,[\alpha^{\,n-k-1}]r_{n,k+1}$.
Starting from $r_{n,n}=1$ and telescoping,
$[\alpha^{\,n-k}]r_{n,k}=\binom{n}{k}/(n+k+2\beta-1)_{n-k}$, whose denominator is
nonzero by \eqref{eq:nondeg} at degree $n$. In
\eqref{eq:connection-coefficients-closed} only the summand $s=0$, which equals
$a_{nm}^{(\alpha,\beta)}=r_{n,m}$, attains the degree $n-m$, the remaining ones having
degree at most $n-m-2s$; this gives the displayed leading coefficient.
\end{proof}

\begin{remark}
For a fixed subdiagonal $n-m=j$, \eqref{eq:connection-coefficients-closed} contains only
$\lfloor j/2\rfloor+1$ terms. This is the reason why the coefficients near the diagonal
admit the simple expression (which we will use later) given next.
\end{remark}

\begin{corollary}
\label{thm:subleading}
With $\beta=\beta_K=-K$ and $\ell=K-n$,
\begin{equation}
c_{n,\,n-1}^{K}=-\,\frac{n\,\alpha_K}{2(\ell+1)}\,,
\qquad
c_{n,\,n-2}^{K}=\frac{n(n-1)\alpha_K^{2}}{4(\ell+1)(2\ell+3)}\,.
\label{eq:subleading}
\end{equation}
\end{corollary}
\begin{proof}
For $m=n-1$ only $s=0$ occurs, so $c^K_{n,n-1}=a^{(\alpha_K,-K)}_{n,n-1}=r_{n,n-1}$,
which by \eqref{eq:subleading-coeff} and $\beta_K+n-1=-(\ell+1)$ equals
$-n\alpha_K/(2(\ell+1))$. For $m=n-2$ there appear two terms, and
\begin{equation}
c_{n,n-2}(\alpha,\beta)=r_{n,n-2}-\frac{n(n-1)}{4\left(\beta+n-\tfrac32\right)}
=\frac{n(n-1)\alpha^{2}}{4(\beta+n-1)(2\beta+2n-3)},
\label{eq:second-subdiagonal}
\end{equation}
by \eqref{eq:subleading-coeff}. Now, putting $\beta=-K$ one gets the second formula.
\end{proof}

Combining \eqref{conn1} with the identity \eqref{eq:master-K} and $\ell_m=K-m$ we
obtain the expansion of the dressed two parameter object along the anti-diagonal,
\begin{equation}\label{eq:wf-transfer}
\sin^K\!\chi\;\Rom{n}{\alpha_K}{\beta_K}(\cot\chi)
=\sum_{m=0}^{n}c_{nm}^{K}\,c(m,\ell_m)\,
\sin^{\ell_m}\!\chi\;\Geg{m}{\ell_m+1}(\cos\chi),
\end{equation}
which is a finite combination of Gegenbauer polynomials with varying degree \emph{and}
varying order along the anti-diagonal $m+\ell_m=K$.

\section{The lattice of dressed functions and its ladder operators}
\label{sec:ladders}

Theorem~\ref{thm:master} puts the classical contiguous relations of the Gegenbauer
polynomials at our disposal for the $\alpha=0$ Romanovski polynomials $\Rom{n}{0}{-K}$,
which, in view of \eqref{eq:harmonics}, we call \emph{hyperspherical}. We next describe
the four elementary moves they generate on the lattice indexed by the
degree $n$ and the order $\ell$. 

Following the standard Physics jargon, in what follows a \emph{dressing} will mean multiplication, after the change of variable $x=\cot\chi$, by a prescribed nonpolynomial factor in $\chi$. Different dressings are used for different purposes. The factor $\sin^n\chi$ in \eqref{eq:master} removes the factor $(1+x^2)^{n/2}$ occurring in the Romanovski--Gegenbauer correspondence, whereas the factor $\sin^K\chi$ in \eqref{eq:master-K} produces the usual hyperspherical expression. For the first order ladder calculus we are going to develop, it is convenient to include one additional factor of $\sin\chi$. We therefore use the term \emph{dressed functions} specifically for the ladder-adapted family 
\begin{equation}\label{eq:fKl}
f_{K\ell}(\chi) = \sin^{\ell+1}\chi\, \Geg{n}{\ell+1}(\cos\chi) = c(n,\ell)^{-1}\sin^{K+1}\chi\, \Rom{n}{0}{-K}(\cot\chi)\,, 
\end{equation}
(with $K=n+\ell$ and $\chi\in(0,\pi)$).
The extra factor $\sin\chi$ is precisely what leads to the first order operators displayed below.

Throughout this section we work with the extension to the boundary of the lattice given
by
\begin{equation}
f_{K\ell}:\equiv0\qquad\text{whenever }n=K-\ell<0,
\label{eq:fKl-boundary}
\end{equation}
which is the convention $\Geg{-1}{\lambda}\equiv0$ of Section~\ref{sec:basic}, read at
the first point outside the lattice. We need it because at the bottom of a ladder the
differential operator annihilates the state, while its nominal coefficient does not
vanish. The functions \eqref{eq:fKl} are, up to normalization, the quasi radial $S^3$
functions $S_{K\ell}$ times an extra $\sin\chi$, and it is the exponent $\ell+1$, rather
than $\ell$, what makes the operators below first order with trigonometric polynomial
coefficients. All the operators act on smooth functions on $(0,\pi)$, where $\cot\chi$
and $1/\sin\chi$ are smooth, and every displayed identity is an identity of functions
there. We decorate the operators with a tilde, in order to distinguish them from the
polynomial families.

The four families are obtained from a classical Gegenbauer relation, by means of a
dressing (the change of variable $u=\cos\chi$ followed by a conjugation with a power of 
$\sin\chi$),
and the transformation is elementary. Let us note that the resulting operators depend on
$n$ and $\ell$ only through $(K,\ell)$, and they are labeled accordingly.

\begin{proposition}[Vertical moves, $\Delta n=\pm1$, $\Delta\ell=0$]
\label{thm:vertical}
Let 
$$
\wt{J}^{\,K}_{\pm}:=\pm\sin\chi\,\tfrac{\dd}{\dd\chi}+(K+1)\cos\chi\,.
$$
Then, with
$K=n+\ell$,
\begin{equation}
\wt{J}^{\,K}_{-}\,f_{K\ell}=(n+2\ell+1)\,f_{(K-1)\ell},
\qquad
\wt{J}^{\,K}_{+}\,f_{K\ell}=(n+1)\,f_{(K+1)\ell}.
\label{eq:vertical}
\end{equation}
At the bottom of the column, $n=0$ and $K=\ell$, the first relation must be read with
\eqref{eq:fKl-boundary}. Note that the coefficient $2\ell+1$ does not vanish, but the
target $f_{(\ell-1)\ell}$ is zero, in agreement with
$\bigl[-\sin\chi\frac{\dd}{\dd\chi}+(\ell+1)\cos\chi\bigr]\sin^{\ell+1}\!\chi=0$.
\end{proposition}
\begin{proof}
These are the two Szeg\H{o} degree shifting relations at fixed order
\cite[\S4.7]{Szego},
\[
\begin{split}
&\bigl[(1-u^2)\tfrac{\dd}{\dd u}+n u\bigr]\Geg{n}{\lambda}= (n+2\lambda-1)\Geg{n-1}{\lambda},
\\
&\bigl[-(1-u^2)\tfrac{\dd}{\dd u}+(n+2\lambda)u\bigr]\Geg{n}{\lambda}= (n+1)\Geg{n+1}{\lambda},
\end{split}
\]
written at $u=\cos\chi$ with $\lambda=\ell+1$ and conjugated by $\sin^{\ell+1}\chi$ (see
\ref{app:conjugation}). The second one has been multiplied by $-1$, in such a
way that both ladder coefficients are positive. Note that the conjugation replaces the
multipliers $n$ and $n+2\ell+2$ by the common value $n+\ell+1=K+1$, and this is the
reason why a single pair of operators covers the whole level.
\end{proof}

\begin{proposition}[Horizontal moves, $\Delta\ell=\pm1$, $\Delta n=0$]
\label{prop:horizontal}
Let
\[
\begin{split}
&\wt{H}^{\,K\ell}_{+}:=-\cos\chi\,\tfrac{\dd}{\dd\chi}+\tfrac{\ell+1}{\sin\chi}+(K+1)\sin\chi\,,\\
&\wt{H}^{\,K\ell}_{-}:=-\cos\chi\,\tfrac{\dd}{\dd\chi}-\tfrac{\ell}{\sin\chi}-(K+1)\sin\chi\, .
\end{split}
\]
Then, with $K=n+\ell$,
\begin{equation}
\begin{split}
& \wt{H}^{\,K\ell}_{+}\,f_{K\ell}=2(\ell+1)\,f_{(K+1)(\ell+1)}\,,
\\
& \wt{H}^{\,K\ell}_{-}\,f_{K\ell}=-\frac{(2\ell+n+1)(2\ell+n)}{2\ell}\,f_{(K-1)(\ell-1)}
\quad(\ell\ge1)\,.
\end{split}
\label{eq:horiz}
\end{equation}
\end{proposition}
\begin{proof}
These come from the two order shifting relations, written in the form used in
\cite{Rosu},
\begin{equation}
\begin{split}
&\bigl[u\tfrac{\dd}{\dd u}+(2\lambda+n)\bigr]\Geg{n}{\lambda}=2\lambda\,\Geg{n}{\lambda+1},
\\
&\bigl[u(1-u^2)\tfrac{\dd}{\dd u}-(2\lambda+n-1-nu^2)\bigr]\Geg{n}{\lambda}
=-\frac{(2\lambda+n-1)(2\lambda+n-2)}{2(\lambda-1)}\Geg{n}{\lambda-1},
\end{split}
\label{eq:geg-order}
\end{equation}
valid for $\lambda>0$ and $\lambda>1$ respectively, and conjugated as in
\ref{app:conjugation}, that is, for the raising move by
$\sin^{\lambda+1}\chi\,(\cdot)\,\sin^{-\lambda}\chi$ and for the lowering move by
$\sin^{\lambda-1}\chi\,(\cdot)\,\sin^{-\lambda}\chi$, with $\lambda=\ell+1$.
\end{proof}

At $\lambda=1$ the prefactor in the second relation of \eqref{eq:geg-order} has a
removable singularity, and after clearing the denominator both sides vanish, because
$\Geg{n}{0}\equiv0$ for $n\ge1$ in the normalization \eqref{eq:gegenbauer-leading}.
Thus, there is no nontrivial order lowering move out of $\ell=0$, a boundary effect
which will appear again below.

\begin{proposition}[In-level moves, $\Delta K=0$]
\label{inlevel:prop:dressed}
Let $\wt{\Lambda}^+_\ell:=-\tfrac{\dd}{\dd\chi}+(\ell+1)\cot\chi$ for $\ell\ge0$ and
$\wt{\Lambda}^-_\ell:=\tfrac{\dd}{\dd\chi}+\ell\cot\chi$ for $\ell\ge1$. Then
\begin{equation}
\begin{split}
&\wt{\Lambda}^+_\ell f_{K\ell}=2(\ell+1)\,f_{K(\ell+1)}\ \ (n\ge0),
\\
&\wt{\Lambda}^-_\ell f_{K\ell}=\frac{(n+1)(n+2\ell+1)}{2\ell}\,f_{K(\ell-1)}\ \ (\ell\ge1),
\end{split}
\label{eq:inlevel-f}
\end{equation}
For $n=0$, hence $K=\ell$, the operator $\wt{\Lambda}^+_\ell$ annihilates
$f_{\ell\ell}=\sin^{\ell+1}\!\chi$. The first relation remains valid, but because the
target $f_{\ell(\ell+1)}$ is zero by \eqref{eq:fKl-boundary}, and not because the
coefficient $2(\ell+1)$ vanishes.
\end{proposition}
\begin{proof}
The underlying Gegenbauer relations are $\tfrac{\dd}{\dd
u}\Geg{n}{\lambda}=2\lambda\,\Geg{n-1}{\lambda+1}$, which is obtained by differentiating
the generating function, and the formula (valid for $\lambda>1$):
\begin{equation}
\tfrac{\dd}{\dd u}\bigl[(1-u^2)^{\lambda-\frac12}\Geg{n}{\lambda}(u)\bigr]
=-\frac{(n+1)(n+2\lambda-1)}{2(\lambda-1)}(1-u^2)^{\lambda-\frac32}\Geg{n+1}{\lambda-1}(u)\,.
\label{eq:inlevel-deriv-comp}
\end{equation}
For \eqref{eq:inlevel-deriv-comp}, let us put
$q:=(1-u^2)\Geg{n}{\lambda}{}'-(2\lambda-1)u\Geg{n}{\lambda}$ and use
\eqref{eq:gegenbauer-ode} in the form
$(1-u^2)\Geg{n}{\lambda}{}''=(2\lambda+1)u\Geg{n}{\lambda}{}'-n(n+2\lambda)\Geg{n}{\lambda}$.
The terms in $u\Geg{n}{\lambda}{}'$ cancel, and
$q'=-(n+1)(n+2\lambda-1)\Geg{n}{\lambda}$, which coincides with the derivative of the
right-hand side by the first relation applied at $(n+1,\lambda-1)$. Evaluating both
sides at $u=1$ with $\Geg{m}{\mu}(1)=\binom{m+2\mu-1}{m}$ shows that the constant of
integration vanishes. Both relations preserve $n+\lambda$, so with $\lambda=\ell+1$ they
are the two in-level moves $(n,\ell)\mapsto(n\mp1,\ell\pm1)$, and the conjugations of
\ref{app:conjugation} give \eqref{eq:inlevel-f}.
\end{proof}

\begin{remark}
\label{inlevel:rem:composition}
Composing the two in-level moves one gets back a scalar,
$\wt{\Lambda}^-_{\ell+1}\wt{\Lambda}^+_\ell f_{K\ell}=n(n+2\ell+2)f_{K\ell}$, with
eigenvalue $n(n+2\ell+2)=(K+1)^2-(\ell+1)^2$.
\end{remark}

We will refer to the boundary row $n=0$, equivalently $\ell=K$, as the \emph{circular row} or \emph{circular sector}. This terminology is borrowed from the physical interpretation of the labels: for fixed principal label $K$, the condition $n=0$ selects the state with maximal angular label $\ell=K$, which is the analogue of a circular state in a central
force problem. In the present paper the expression is used primarily as a convenient name for this boundary of the lattice. Algebraically, the corresponding Romanovski polynomial has degree zero, $ \Rom{0}{\alpha_K}{\beta_K}\equiv1$, so the polynomial part is trivial and all dependence on $K$ is carried by the dressing. Thus, for the undeformed dressed functions, we have $ f_{KK}(\chi)=\sin^{K+1}\chi$, while for the degree dependent two-parameter family, $F_{KK}(\chi) = e^{-\alpha_K\chi/2}\sin^{K+1}\chi$.

The \emph{diagonal} operators are the horizontal ones restricted to the bottom row
$n=0$, on which $f_{KK}=\sin^{K+1}\chi$, because $\Geg{0}{\lambda}\equiv1$. We write
$\wt{O}^{\,\pm}_{K}:=\wt{H}^{\,KK}_{\pm}$, so that
\begin{equation}
\begin{split}
&\wt{O}^{\,+}_{K}f_{KK}=2(K+1)f_{(K+1)(K+1)}\quad(K\ge0),
\\
&\wt{O}^{\,-}_{K}f_{KK}=-(2K+1)f_{(K-1)(K-1)}\quad(K\ge1)\,.
\end{split}
\label{eq:Odiag}
\end{equation}
Note that the polynomial content of the circular row is trivial. What the diagonal
operators construct, starting from the ground function $f_{00}=\sin\chi$, is the family of
level dressings $\sin^{K+1}\chi$.

\section{Recurrence relations at fixed parameters}
\label{sec:recurrences}

Now we leave the trigonometric picture and work directly with the polynomials. The
relations of this section hold at fixed $(\alpha,\beta)$ and they are, as we recalled in
the Introduction, the specialization to \eqref{eq:RomanovskiDE} of the classical
structure relations of the hypergeometric type families \cite[Ch.~1]{Nikiforov},
\cite{Koepf,Ismail}. Our aim is to write them out with all the constants, and this is
what makes them usable in the degree dependent setting, because, as they do not move
$(\alpha,\beta)$, they act inside a level multiplet. By
Proposition~\ref{prop:jacobi-identity} they are equally the Jacobi structure relations
transported by $z=\ii x$; what is specific here is the real monic normalization and its
use inside a level. Let us write
\begin{equation}
\begin{split}
&\wt L_{\mu}:=(1+x^2)\frac{\dd^2}{\dd x^2}+(2\beta x+\alpha)\frac{\dd}{\dd x}-\mu,
\\
&\varrho_m:=m(2\beta+m-1),\quad
\wt T_{a,b}:=(1+x^2)\frac{\dd}{\dd x}+a\,x+b ,
\end{split}
\label{eq:inlevel-L}
\end{equation}
so that \eqref{eq:RomanovskiDE} reads $\wt L_{\varrho_n}\Rom{n}{\alpha}{\beta}=0$.

\begin{lemma}
\label{inlevel:lem:intertwine}
For any constants $a,b,\mu$ and any smooth $f$,
\begin{align}
\wt L_{\mu}\bigl[\wt T_{a,b}f\bigr]-\wt T_{a+2,b}\bigl[\wt L_{\varrho_n}f\bigr]
&=\bigl[2(a+1-\beta)-(\mu-\varrho_n)\bigr](1+x^2)f' \nonumber\\
&\quad+\Bigl\{\bigl[a(2\beta-\mu)+(a+2)\varrho_n\bigr]x
+a\alpha-b(\mu-\varrho_n)\Bigr\}f .
\label{eq:inlevel-residual}
\end{align}
In particular, $\wt L_{\mu}\wt T_{a,b}=\wt T_{a+2,b}\wt L_{\varrho_n}$ as differential
operators if and only if
\begin{equation}
\begin{split}
&\mathrm{(C1)}\ \mu=\varrho_n+2(a+1-\beta),
\\
&\mathrm{(C2)}\ \varrho_n=a(a+1-2\beta),
\\
&\mathrm{(C3)}\ a\alpha=2b(a+1-\beta).
\end{split}
\label{eq:inlevel-C}
\end{equation}
Here condition \textup{(C2)} is the quadratic $a^2+(1-2\beta)a-n(2\beta+n-1)=0$, with
roots
\begin{equation}
a=-n \qquad\text{and}\qquad a=2\beta+n-1 ,
\label{eq:inlevel-roots}
\end{equation}
which by \textup{(C1)} give $\mu=\varrho_{n-1}$ and $\mu=\varrho_{n+1}$ respectively;
and, when $a+1-\beta\neq0$, \textup{(C3)} determines $b=\frac{a\alpha}{2(a+1-\beta)}$.
\end{lemma}
\begin{proof}
With $g:=\wt T_{a,b}f$ one has $g'=(1+x^2)f''+[(a+2)x+b]f'+af$ and
$g''=(1+x^2)f'''+[(a+4)x+b]f''+(2a+2)f'$, while with $h:=\wt L_{\varrho_n}f$ one has
$h'=(1+x^2)f'''+[(2+2\beta)x+\alpha]f''+(2\beta-\varrho_n)f'$. Expanding $\wt L_\mu g$
and $\wt T_{a+2,b}h$ and subtracting, the $f'''$ terms agree, the $f''$ terms agree
because $(a+4)x+b+2\beta x+\alpha=(2+2\beta)x+\alpha+(a+2)x+b$, and the products
$(2\beta x+\alpha)[(a+2)x+b]f'$ cancel, so what is left are the stated $f'$ and $f$
coefficients. The vanishing of the $f'$ coefficient is (C1). Substituting (C1) into the
$x-$part of the $f$ coefficient gives $a(4\beta-2a-2)+2\varrho_n=0$, that is, (C2), and
the constant part is (C3). Now, the roots \eqref{eq:inlevel-roots} have sum $2\beta-1$
and product $-n(2\beta+n-1)$, and (C1) gives then
$\mu=(n-1)(2\beta+n-2)=\varrho_{n-1}$ and $\mu=(n+1)(2\beta+n)=\varrho_{n+1}$.
\end{proof}

\begin{theorem}
\label{inlevel:thm:lower}
Let $n\ge1$, and assume \eqref{eq:nondeg} at the degrees $n-1$ and $n$. Then, in the
monic normalization,
\begin{equation}\label{eq:inlevel-low}
\Bigl[(1+x^2)\frac{\dd}{\dd x}-n\,x+r_{n,n-1}\Bigr]\Rom{n}{\alpha}{\beta}(x)
=\gamma^{-}_{n}\;\Rom{n-1}{\alpha}{\beta}(x)\,,
\end{equation}
where $r_{n,n-1}=\dfrac{n\alpha}{2(\beta+n-1)}$ is the subleading coefficient
\eqref{eq:subleading-coeff}, and
\begin{equation}
\gamma^{-}_{n}=\frac{n\,(2\beta+n-2)}{2\beta+2n-3}
\biggl[\,1+\frac{\alpha^{2}}{4(\beta+n-1)^{2}}\,\biggr]\,.
\label{eq:inlevel-low-const}
\end{equation}
\end{theorem}
\begin{proof}
Let us set $\wt M_{n}:=\wt T_{-n,\,r_{n,n-1}}$ and $P:=\wt M_n\Rom{n}{\alpha}{\beta}$.
The coefficient of $x^{n+1}$ in $P$ is $n-n=0$, and the one of $x^{n}$ is
$(n-1)r_{n,n-1}-nr_{n,n-1}+r_{n,n-1}=0$, this last cancellation identifying $r_{n,n-1}$
uniquely. Now, Lemma~\ref{inlevel:lem:intertwine} with $a=-n$, $b=r_{n,n-1}$ and
$\mu=\varrho_{n-1}$ (a pair which verifies (C3), because $a+1-\beta=1-n-\beta\neq0$ by
\eqref{eq:nondeg} at degree $n$ with $k=n-1$, and
$\frac{a\alpha}{2(a+1-\beta)}=\frac{n\alpha}{2(\beta+n-1)}=r_{n,n-1}$) gives us $\wt
L_{\varrho_{n-1}}P=\wt T_{2-n,\,r_{n,n-1}}\wt L_{\varrho_n}\Rom{n}{\alpha}{\beta}=0$.
Hence $P$ is a polynomial solution of $\wt L_{\varrho_{n-1}}y=0$ of degree $\le n-1$, so
that $P=\gamma\Rom{n-1}{\alpha}{\beta}$ by uniqueness. Finally, $\gamma$ is the
coefficient of $x^{n-1}$ in $P$, that is, $\gamma=n-2r_{n,n-2}+r_{n,n-1}^{2}$. Inserting
\eqref{eq:subleading-coeff} and reducing over the common denominator
$4(\beta+n-1)^{2}(2\beta+2n-3)$, the $\alpha-$independent part of the numerator is
$4n(\beta+n-1)^{2}(2\beta+n-2)$, while the coefficient of $\alpha^{2}$ is
$n(2\beta+n-2)$, and this gives \eqref{eq:inlevel-low-const}.
\end{proof}

At $n=1$ the quotient $(2\beta+n-2)/(2\beta+2n-3)=(2\beta-1)/(2\beta-1)$ must be 
understood as
prescribed in Remark~\ref{rem:cancellation}, where it equals $1$, so that
$\gamma^{-}_{1}=1+\alpha^{2}/(4\beta^{2})$. The same value comes directly from the
computation of coefficients in the proof, which at $n=1$ reads
$1-2r_{1,-1}+r_{1,0}^{2}=1+r_{1,0}^{2}$, by the convention $r_{n,j}=0$ for $j<0$.
Moreover, inside the affine family of operators $(1+x^2)\frac{\dd}{\dd x}-n\,x+b$,
$b\in\R$, the value $b=r_{n,n-1}$ is the only one for which the image of
$\Rom{n}{\alpha}{\beta}$ has degree $\le n-1$. 

\begin{theorem}
\label{inlevel:thm:raise}
Let $n\ge0$, assume \eqref{eq:nondeg} at the degrees $n$ and $n+1$, and assume
$2\beta+2n-1\neq0$. Then
\begin{equation}
\begin{split}
&\Bigl[-(1+x^2)\frac{\dd}{\dd x}-(2\beta+n-1)\Bigl(x+\frac{\alpha}{2(\beta+n)}\Bigr)\Bigr]
\Rom{n}{\alpha}{\beta}(x)=\gamma^{+}_{n}\;\Rom{n+1}{\alpha}{\beta}(x)\,,
\\ 
&\gamma^{+}_{n}:=-(2\beta+2n-1)\,,
\end{split}
\label{eq:inlevel-raise}
\end{equation}
the proportionality constant being thus independent of $\alpha$.
\end{theorem}
\begin{proof}
Set $\wt N_{n}:=-\wt T_{\,2\beta+n-1,\;(2\beta+n-1)\alpha/(2(\beta+n))}$ and
$Q:=\wt N_n\Rom{n}{\alpha}{\beta}$. Here $a+1-\beta=\beta+n\neq0$ by \eqref{eq:nondeg}
at degree $n+1$ with $k=n$, and (C3) holds by construction, so
Lemma~\ref{inlevel:lem:intertwine} with the second root of \eqref{eq:inlevel-roots}
gives $\wt L_{\varrho_{n+1}}Q=0$. The coefficient of $x^{n+1}$ in $Q$ equals
$-n-(2\beta+n-1)=-(2\beta+2n-1)$, which is nonzero by hypothesis. Thus $Q$ has degree
exactly $n+1$ and, by uniqueness,
$Q=-(2\beta+2n-1)\Rom{n+1}{\alpha}{\beta}$.
\end{proof}
\begin{remark}
The hypothesis
$2\beta+2n-1\neq0$ (which is automatic for $n\ge1$, being
\eqref{eq:nondeg} at degree $n+1$ with $k=n-1$) can not be dropped at $n=0$. 
Indeed, for $\beta=\tfrac12$ both sides
of \eqref{eq:inlevel-raise} vanish identically, so the relation degenerates to $0=0$ and
it does not produce $\Rom{1}{\alpha}{\beta}$.
\end{remark}

Adding \eqref{eq:inlevel-low} and \eqref{eq:inlevel-raise} eliminates the derivative and
gives the three-term recurrence announced in
Proposition~\ref{prop:3TR-statement}.

\begin{corollary}\label{cor:3TR}
Let $n\ge1$ and let the hypotheses of Theorems~\textup{\ref{inlevel:thm:lower}} and
\textup{\ref{inlevel:thm:raise}} hold. Then \eqref{eq:3TR} holds with
\begin{equation}
\begin{split}
&b_n=\frac{1}{2\beta+2n-1}\Bigl[r_{n,n-1}-\frac{(2\beta+n-1)\alpha}{2(\beta+n)}\Bigr]
=\frac{\alpha(1-\beta)}{2(\beta+n-1)(\beta+n)},
\\
&c_n=\frac{\gamma^{-}_{n}}{\gamma^{+}_{n}},
\end{split}
\label{eq:3TR-from-ladders}
\end{equation}
which are the expressions \eqref{eq:3TR-coeffs}. For $n=0$, if $\beta\neq0$ then
\eqref{eq:3TR} holds with $b_0=-\alpha/(2\beta)$, $c_0=0$ and
$\Rom{-1}{\alpha}{\beta}=0$.
\end{corollary}
\begin{proof}
Let $n\ge1$. The sum of the left-hand sides of \eqref{eq:inlevel-low} and
\eqref{eq:inlevel-raise} is
$\bigl[-(2\beta+2n-1)x+r_{n,n-1}-\frac{(2\beta+n-1)\alpha}{2(\beta+n)}\bigr]
\Rom{n}{\alpha}{\beta}$, and the sum of the right-hand sides is
$\gamma^-_n\Rom{n-1}{\alpha}{\beta}+\gamma^+_n\Rom{n+1}{\alpha}{\beta}$. Dividing by
$\gamma^+_n=-(2\beta+2n-1)$, which is nonzero by \eqref{eq:nondeg} at degree $n+1$ with
$k=n-1$, we get \eqref{eq:3TR}. As for $b_n$, the bracket equals
$\frac{\alpha}{2}\cdot\frac{n(\beta+n)-(2\beta+n-1)(\beta+n-1)}{(\beta+n-1)(\beta+n)}$ and
the numerator factors as $(1-\beta)(2\beta+2n-1)$.

The case $n=0$ lies outside the scope of Theorem~\ref{inlevel:thm:lower}, and the two
expressions in \eqref{eq:3TR-from-ladders} involve the symbols $r_{0,-1}$ and
$\gamma^-_0$, which are not defined. So, it must be verified directly. By
\eqref{eq:CoeffRecursionSolved}, $\Rom{0}{\alpha}{\beta}=1$ and
$\Rom{1}{\alpha}{\beta}(x)=x+\frac{\alpha}{2\beta}$, so that
$x\,\Rom{0}{\alpha}{\beta}=x=\Rom{1}{\alpha}{\beta}(x)-\frac{\alpha}{2\beta}$, and this
is \eqref{eq:3TR} with $b_0=-\alpha/(2\beta)$ and $c_0=0$.
\end{proof}

Specializing to the degree dependent parameters, $\beta=\beta_K=-K$ and
$\alpha=\alpha_K$, and writing everything in terms of $\ell=K-n$, one has
$\beta+n-1=-(\ell+1)$, $\beta+n=-\ell$, $2\beta+n-1=-(n+2\ell+1)$,
$2\beta+n-2=-(n+2\ell+2)$, $2\beta+2n-1=-(2\ell+1)$ and $2\beta+2n-3=-(2\ell+3)$, and we
arrive at the following result.

\begin{corollary}
\label{inlevel:cor:physical}
Let $K\ge1$, $\ell=K-n$, $\beta_K=-K$ and $\alpha_K\in\R$ arbitrary. Then,
for $1\le n\le K$,
\begin{equation}
\Bigl[(1+x^2)\frac{\dd}{\dd x}-n\,x-\frac{n\,\alpha_K}{2(\ell+1)}\Bigr]
\Rom{n}{\alpha_K}{\beta_K}
=\frac{n(n+2\ell+2)}{2\ell+3}\biggl[1+\frac{\alpha_K^{2}}{4(\ell+1)^{2}}\biggr]
\Rom{n-1}{\alpha_K}{\beta_K}\,, 
\label{eq:inlevel-low-phys}
\end{equation}
and, for $0\le n\le K-1$,
\begin{equation}
\Bigl[-(1+x^2)\frac{\dd}{\dd x}+(n+2\ell+1)\Bigl(x-\frac{\alpha_K}{2\ell}\Bigr)\Bigr]
\Rom{n}{\alpha_K}{\beta_K}
=(2\ell+1)\,\Rom{n+1}{\alpha_K}{\beta_K}\,.
\label{eq:inlevel-raise-phys}
\end{equation}
\end{corollary}

Both sides of each relation carry the same pair of parameters $(\alpha_K,\beta_K)$, so
inside every level multiplet these are first order differential recurrences in the
degree alone. The additive constant (the potential term, in physical terminology) in \eqref{eq:inlevel-low-phys} is
$-n\alpha_K/(2(\ell+1))=c^K_{n,n-1}$, which is the subleading connection coefficient of
Corollary~\ref{thm:subleading}. This is not a coincidence, because by
Theorem~\ref{inlevel:thm:lower} that constant equals $r_{n,n-1}$, and $r_{n,n-1}$ is
precisely what Corollary~\ref{thm:subleading} computes. The restriction $\ell\ge1$ in
\eqref{eq:inlevel-raise-phys} says that no further in-level raising is possible from the
top corner $n=K$ of a level, as it should be, because $\Rom{K+1}{\alpha_K}{-K}$ would
exceed the multiplet.

The two relations above preserve $(\alpha_K,\beta_K)$, and therefore they act inside a
level multiplet. It is worth recording where they stand with respect to the tables of
\cite[\S9.9]{KLS}. Under the dictionary \eqref{eq:KLS-dictionary} those tables list the
recurrence \cite[(9.9.4)]{KLS}, which is \eqref{eq:3TR}, together with a forward and a
backward shift operator, \cite[(9.9.6)]{KLS} and \cite[(9.9.7)]{KLS}. Both of them move
$N$, and hence $\beta$: the first one is the Appell type relation
$\frac{\dd}{\dd x}\Rom{n}{\alpha}{\beta}=n\,\Rom{n-1}{\alpha}{\beta+1}$, while the second
one reads, in the present variables,
\begin{equation}
(1+x^{2})\frac{\dd}{\dd x}\Rom{n}{\alpha}{\beta}
+\bigl[\alpha+2(\beta-1)x\bigr]\Rom{n}{\alpha}{\beta}
=(2\beta+n-2)\,\Rom{n+1}{\alpha}{\beta-1}\,.
\label{eq:KLS-backward}
\end{equation}
Neither of them preserves the pair of parameters, so neither of them acts inside a level
multiplet, and the pair \eqref{eq:inlevel-low}, \eqref{eq:inlevel-raise} --- which does
--- is not among them. It is obtained instead from the structure relation, as we recalled
at the beginning of this section.

\section{Trigonometric form, level chains, and an iterated raising formula}
\label{sec:trig}

Now we return to the trigonometric variable and restrict ourselves to the specific
form $\alpha_K=c/(K+1)$. Let us define, for $K\ge0$ and $0\le\ell\le K$ (with
$n=K-\ell$),
\begin{equation}
F_{K\ell}(\chi):=e^{-\alpha_K\chi/2}\,\sin^{K+1}\!\chi\;\Rom{n}{\alpha_K}{\beta_K}(\cot\chi),
\qquad \chi\in(0,\pi),
\label{eq:inlevel-F}
\end{equation}
with the same boundary convention as in \eqref{eq:fKl-boundary}, namely
$F_{K\ell}:\equiv0$ whenever $K-\ell<0$.
These are a natural deformation of the dressed functions \eqref{eq:fKl}. 
Indeed, multiplying \eqref{eq:master} by $\sin^{\ell+1}\chi$ gives
$\sin^{K+1}\chi\,\Rom{n}{0}{\beta_K}(\cot\chi)=c(n,\ell)f_{K\ell}(\chi)$, so at
$\alpha_K=0$ one has $F_{K\ell}=c(n,\ell)f_{K\ell}$, while for $\alpha_K\neq0$ the extra
factor $e^{-\alpha_K\chi/2}$ still depends only on the level, and thus it is common to
all the members of one multiplet. The essential feature here is that the whole dressing
$e^{-\alpha_K\chi/2}\sin^{K+1}\chi$ is a function of $K$ alone.

\begin{proposition}
\label{inlevel:prop:chi}
For $\ell\ge0$ and $\ell\ge1$ respectively, let us define the $\ell-$indexed and
\emph{$K-$independent} operators
\begin{equation}
\wt A^{\,-}_{\ell}:=\frac{\dd}{\dd\chi}-(\ell+1)\cot\chi+\frac{c}{2(\ell+1)}\,,
\quad
\wt A^{\,+}_{\ell}:=-\frac{\dd}{\dd\chi}-\ell\cot\chi+\frac{c}{2\ell}\,.
\label{eq:inlevel-A}
\end{equation}
Then
\begin{align}
\wt A^{\,-}_{\ell}F_{K\ell}
&=-\frac{n(n+2\ell+2)}{2\ell+3}\biggl[1+\frac{\alpha_K^{2}}{4(\ell+1)^{2}}\biggr]
F_{K(\ell+1)},
&&0\le\ell\le K,
\label{eq:inlevel-A-act}\\
\wt A^{\,+}_{\ell}F_{K\ell}&=-(2\ell+1)\,F_{K(\ell-1)},
&&1\le\ell\le K .
\label{eq:inlevel-A-act2}
\end{align}
\end{proposition}
\begin{proof}
Let us write $F_{K\ell}=e^{-\alpha_K\chi/2}\sin^{K+1}\chi\,R(\cot\chi)$ with
$R=\Rom{n}{\alpha_K}{\beta_K}$, and use the identity
$\frac{\dd}{\dd\chi}[R(\cot\chi)]=-(1+x^2)R'|_{x=\cot\chi}$ together with
$\frac{\dd}{\dd\chi}[e^{-\alpha_K\chi/2}\sin^{K+1}\chi]
=[-\tfrac{\alpha_K}{2}+(K+1)\cot\chi]e^{-\alpha_K\chi/2}\sin^{K+1}\chi$. For
$\wt A^{\,-}_\ell$ the resulting bracket is
$-(1+x^2)R'+nxR+\bigl(\tfrac{c}{2(\ell+1)}-\tfrac{\alpha_K}{2}\bigr)R$, and since
$c=(K+1)\alpha_K$,
$\tfrac{c}{2(\ell+1)}-\tfrac{\alpha_K}{2}=\tfrac{n\alpha_K}{2(\ell+1)}$, so that the
bracket is minus the left-hand side of \eqref{eq:inlevel-low-phys}. In the same way, for
$\wt A^{\,+}_\ell$, with $(K+1)+\ell=n+2\ell+1$ and
$\tfrac{\alpha_K}{2}+\tfrac{c}{2\ell}=\tfrac{(n+2\ell+1)\alpha_K}{2\ell}$, the bracket
is minus the left-hand side of \eqref{eq:inlevel-raise-phys}. Now, the targets
$F_{K(\ell\pm1)}$ carry the same level, hence the same dressing, and both
\eqref{eq:inlevel-A-act} and \eqref{eq:inlevel-A-act2} follow.
\end{proof}
\begin{remark}
In \eqref{eq:inlevel-A-act} the case $n=0$, that is, $\ell=K$, involves three separate
facts, each one of which would by itself make the identity trivial. The scalar
coefficient contains the factor $n$, so it vanishes; the target $F_{K(K+1)}$ vanishes by
the boundary convention; and $\wt A^{\,-}_{K}$ annihilates $F_{KK}$. Thus, the identity
holds at the boundary of the lattice (see Figure~\ref{fig:ladder}).
\end{remark}

The $K-$independence of \eqref{eq:inlevel-A} is the direct consequence of the choice
$\alpha_K=c/(K+1)$; in this case we obtain a single pair of operators, indexed by the angular label alone, that traverses every level multiplet in the lattice of dressed
functions. Note also that $\wt A^{\,-}_{\ell}$ is the
derivative operator conjugated by the seed of the $\ell-$th column. Indeed, with
$F_{\ell\ell}=e^{-c\chi/(2(\ell+1))}\sin^{\ell+1}\chi$,
\[
F_{\ell\ell}\circ\frac{\dd}{\dd\chi}\circ F_{\ell\ell}^{-1}
=\frac{\dd}{\dd\chi}-\frac{F_{\ell\ell}'}{F_{\ell\ell}}
=\frac{\dd}{\dd\chi}-(\ell+1)\cot\chi+\frac{c}{2(\ell+1)}
=\wt A^{\,-}_{\ell}\,.
\]

Since $\Rom{0}{\alpha}{\beta}\equiv1$ for every pair of parameters, the diagonal step in
such a chain only updates the dressing, and all the polynomial information is produced
by the in-level raisings. Iterating them we get an explicit generation formula.

\begin{theorem}
\label{inlevel:thm:product}
Let $\beta\in\R$ satisfy
\begin{equation}
2\beta+r-1\neq0\qquad\text{for }r=0,1,\dots,2n-1,
\label{eq:inlevel-product-hyp}
\end{equation}
and set $\wt N_{j}:=-(1+x^2)\frac{\dd}{\dd x}-(2\beta+j-1)\bigl(x+\frac{\alpha}{2(\beta+j)}\bigr)$.
Then, in the monic normalization,
\begin{equation}
\Rom{n}{\alpha}{\beta}(x)=\Biggl[\prod_{j=0}^{n-1}\bigl(-(2\beta+2j-1)\bigr)\Biggr]^{-1}
\wt N_{n-1}\,\wt N_{n-2}\cdots\wt N_{1}\,\wt N_{0}\;\mathbf{1}\,,
\label{eq:inlevel-product-general}
\end{equation}
the empty product for $n=0$ being $1$. For $\beta=\beta_K=-K$ and $\alpha=\alpha_K$ one
has $-(2\beta+j-1)=2K-j+1$ and $x+\frac{\alpha}{2(\beta+j)}=x-\frac{\alpha_K}{2(K-j)}$,
in such a way that every member of the level multiplet is generated from $\mathbf 1$ by
at most $K$ first order operations.
\end{theorem}
\begin{proof}
It is an induction on $n$, the inductive step being Theorem~\ref{inlevel:thm:raise} at
degree $j=n-1$, whose proportionality constant is $-(2\beta+2n-3)$.
\end{proof}
The iteration adds no existence statement beyond the repeated application of
Theorem~\ref{inlevel:thm:raise}.

The hypothesis \eqref{eq:inlevel-product-hyp} guarantees the regularity of all the denominators and normalizing constants entering the construction. 
Its odd values $r=2j+1$ guarantee
$\beta+j\neq0$, so that the additive constants of the $\wt N_{j}$ are well defined; 
its even
values $r=2j$ guarantee that no normalizing factor vanishes; and together they contain
the nondegeneracy conditions \eqref{eq:nondeg} at every intermediate degree $m\le n$.
Altogether, \eqref{eq:inlevel-product-hyp} excludes only the finite set
$\bigl\{\tfrac12,0,-\tfrac12,-1,\dots,-(n-1)\bigr\}$ of values of $\beta$. In the level
case the first failure occurs at $r=2K+1$, that is, $\beta+K=0$, which enters the range
$0\le r\le2n-1$ precisely when $n=K+1$. Thus, the construction stops exactly at $n=K$,
and this explains the finite size of the multiplet. In the trigonometric picture the
same iteration reads, for $0\le\ell\le K$,
\[
\begin{split}
F_{K\ell}=& \Biggl[\prod_{\ell'=\ell+1}^{K}\bigl(-(2\ell'+1)\bigr)\Biggr]^{-1}
\wt A^{\,+}_{\ell+1}\,\wt A^{\,+}_{\ell+2}\cdots\wt A^{\,+}_{K}\;F_{KK},
\\
F_{KK}=& e^{-\alpha_K\chi/2}\sin^{K+1}\!\chi\,.
\end{split}
\]

\begin{remark}
\label{rem:rodrigues-comparison}
Both \eqref{eq:inlevel-product-general} and the Rodrigues representation
\eqref{eq:RomanovskiRodrigues} are $n-$step differential constructions of the same
polynomial, but they are of a different nature. The Rodrigues formula differentiates
$w^{\alpha,\beta}(x)(1+x^{2})^{n}$ once, to order $n$, and its intermediate derivatives
are not members of the family. On the other hand, \eqref{eq:inlevel-product-general} is
an ordered composition of $n$ first order operators, each one of which maps the monic
family into itself, in such a way that every intermediate object $\wt N_{j-1}\cdots\wt
N_{0}\mathbf 1$ is again a multiple of some $\Rom{j}{\alpha}{\beta}$, the monic
normalization is preserved by the explicit normalizing product, and the construction
specializes without any change to the degree dependent parameters.
\end{remark}

\begin{remark}
\label{rem:factorization}
The pair \eqref{eq:inlevel-A} is a Schr\"odinger factorization. Substituting
$x=\cot\chi$ in \eqref{eq:RomanovskiDE} at
the parameters \eqref{eq:degdep}, and using $\varrho_n=-K(K+1)+\ell(\ell+1)$ together
with $(K+1)\alpha_K=c$, one gets
\[
H_\ell F_{K\ell}=E_K F_{K\ell},\qquad
H_\ell=-\frac{\dd^{2}}{\dd\chi^{2}}+\ell(\ell+1)\csc^{2}\chi-c\cot\chi,
\]
with $E_K=(K+1)^{2}-c^{2}/\bigl(4(K+1)^{2}\bigr)$, and we write $E_\ell$ for the same
expression at $K=\ell$.
Putting $W_\ell:=(\ell+1)\cot\chi-\frac{c}{2(\ell+1)}$, the operators
\eqref{eq:inlevel-A} read $\wt A^{\,-}_{\ell}=\frac{\dd}{\dd\chi}-W_\ell$ and
$\wt A^{\,+}_{\ell+1}=-\frac{\dd}{\dd\chi}-W_\ell$, the second being the formal adjoint
of the first (leaving aside any question about operator domains), and multiplying out,
\[
\wt A^{\,+}_{\ell+1}\wt A^{\,-}_{\ell}=H_\ell-E_\ell,
\qquad
\wt A^{\,-}_{\ell}\wt A^{\,+}_{\ell+1}=H_{\ell+1}-E_\ell .
\]
Accordingly the product of the two constants in \eqref{eq:inlevel-A-act} and
\eqref{eq:inlevel-A-act2} is
$E_K-E_\ell=n(n+2\ell+2)\bigl[1+\alpha_K^{2}/(4(\ell+1)^{2})\bigr]$. Here $H_\ell$
involves $c$ and not $\alpha_K$ separately, which is the $K-$independence of
\eqref{eq:inlevel-A}.
\end{remark}

\section{$\mathrm{su}(1,1)$ structures on the lattice of dressed functions}
\label{sec:visualization}

Figure~\ref{fig:ladder} shows the lattice of dressed functions in the $(\ell,n)$ plane,
together with the four families of moves defined in Section~\ref{sec:ladders}. 
Our purpose now is to examine which of them integrate to a lowest weight 
$\mathrm{su}(1,1)$ module.

\begin{figure}[htbp]
\centering
\begin{tikzpicture}[scale=0.70, every node/.style={font=\small}]
  \pgfmathsetmacro{\hs}{1.6}
  \pgfmathsetmacro{\vs}{1.05}
  \foreach \Kc in {1,...,6}{
    \pgfmathtruncatemacro{\lmax}{min(\Kc,5)}
    \pgfmathtruncatemacro{\lmin}{max(0,\Kc-4)}
    \draw[gray, dashed]
      ($(\lmin*\hs,{(\Kc-\lmin)*\vs})$) -- ($(\lmax*\hs,{(\Kc-\lmax)*\vs})$);
  }
  \foreach \l in {0,...,5}{
    \foreach \nn in {0,...,4}{
      \node[circle, fill=black!12, draw=black!45, inner sep=1.7pt,
            minimum size=5pt] (p\l\nn) at (\l*\hs, \nn*\vs) {};
    }
  }
  \foreach \l in {0,...,5}{
    \node[circle, fill=orange!70, draw=orange!90, inner sep=1.7pt,
          minimum size=5pt] at (\l*\hs, 0) {};
  }
  \foreach \l in {0,1,2}{
    \foreach \nn in {0,1,2,3}{
      \pgfmathtruncatemacro{\np}{\nn+1}
      \draw[blue!80,->,>=stealth,thick] ($(\l*\hs,\nn*\vs)+(0,0.16)$) -- ($(\l*\hs,\np*\vs)-(0,0.16)$);
    }
  }
  \foreach \nn in {0,1,2}{
    \foreach \l in {0,1,2,3}{
      \pgfmathtruncatemacro{\lp}{\l+1}
      \draw[purple!80,->,>=stealth,thick] ($(\l*\hs,\nn*\vs)+(0.16,0)$) -- ($(\lp*\hs,\nn*\vs)-(0.16,0)$);
    }
  }
  \foreach \l in {0,...,5}{\node[below=8pt,font=\footnotesize] at (\l*\hs,0) {$\ell=\l$};}
  \foreach \nn in {0,...,4}{\node[left=8pt,font=\footnotesize] at (0,\nn*\vs) {$n=\nn$};}
  \begin{scope}[shift={(-1.2,-1.45)}]
    \draw[blue!80,->,>=stealth,thick] (0,-0.5) -- (0.9,-0.5)
      node[right,font=\footnotesize,black] {$\wt{J}^{\,K}_{\pm}$: $\Delta n=\pm1$, $\Delta\ell=0$};
    \draw[purple!80,->,>=stealth,thick] (0,-1.25) -- (0.9,-1.25)
      node[right,font=\footnotesize,black] {$\wt{H}^{\,K\ell}_{\pm}$: $\Delta\ell=\pm1$, $\Delta n=0$};
    \node[circle, fill=orange!70, draw=orange!90, inner sep=1.7pt, minimum size=5pt] at (0.15,-2) {};
    \node[right=2pt,font=\footnotesize,black] at (0.3,-2) {circular row $n=0$: $\wt{O}^{\,\pm}_{K}=\wt{H}^{\,KK}_{\pm}$};
    \draw[gray,dashed] (0,-2.75) -- (0.9,-2.75)
      node[right,font=\footnotesize,black] {levels $K=n+\ell$: in-level $\wt{\Lambda}^\pm_\ell$};
  \end{scope}
\end{tikzpicture}
\caption[The lattice of dressed functions]{The lattice of dressed functions $f_{K\ell}$ in the $(\ell,n)$ plane. Each
vertical column (fixed $\ell$) carries an irreducible $\mathrm{su}(1,1)$ lowest weight
module of Bargmann index $k=\ell+1$ (Theorem~\ref{thm:su11-vertical}). Every row
(fixed $n$) can be organized, after a boundary prescription at $\ell=0$ and a
label dependent rescaling, as the module of index $k$, for any $k>0$
(Proposition~\ref{prop:rows}); the circular row $n=0$ with $k=1$ is
Proposition~\ref{thm:su11}. The dashed
anti-diagonals are the levels.}
\label{fig:ladder}
\end{figure}
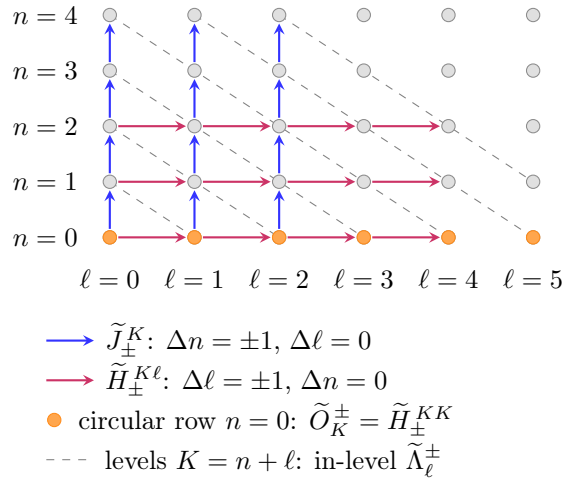

Before going on, let us fix the sense in which label dependent families of operators are
composed. A family $\{A_{\lambda}\}$, where $A_{\lambda}$ is applied to the object
carrying the label $\lambda$, determines a global linear operator $A$ on the linear span
of those objects through $A\,v_\lambda:=A_\lambda v_\lambda$, extended by linearity, and
the products and commutators are those of the global operators, provided the
$v_\lambda$ are linearly independent. In other words, the
global operators are differential operators whose constant coefficients are functions of
the number operator. We write the global operators without tilde, reserving the tilded
symbols for their differential representatives.

We would like to stress two points. First, a family of differential expressions defines a
global operator on the span only where its image stays inside that span; at a boundary
label this may fail, and there the action of the global operator has to be prescribed
abstractly, instead of being read off from the differential expression. This happens
below, at the bottom of every row. Second, the linear independence just required does
\emph{not} hold on the whole lattice. Indeed, since $\Geg{n}{1}=U_n$ and
$\sin\chi\,U_n(\cos\chi)=\sin((n+1)\chi)$, one has $f_{20}=\sin3\chi$ and
$3f_{00}-4f_{22}=3\sin\chi-4\sin^{3}\chi=\sin3\chi$, that is,
\[
f_{20}=3f_{00}-4f_{22}\,,
\]
while the labelled vertical raising \eqref{eq:vertical} would send the zero function
$f_{20}-3f_{00}+4f_{22}$ to $3f_{30}-3f_{10}+4f_{32}=3\sin4\chi-3\sin2\chi+24\sin^{3}\chi\cos\chi
=3\sin2\chi\neq0$. Hence \eqref{eq:vertical} does not define a single linear operator on
the span of all the $f_{K\ell}$. Accordingly, each construction below is carried out
along one column or one row, where the linear independence is checked in each case; a
simultaneous construction on the whole lattice would have to be made on the abstract
space spanned by the labels, which maps onto the span of the $f_{K\ell}$, but not
injectively.

\subsection{Vertical columns}

Let us fix $\ell\ge0$ and consider the column $\{f_{K\ell}\}_{n\ge0}$, $K=n+\ell$. Here
the first order operators already close the algebra, with no rescaling beyond the
normalization of the states. Indeed, by Proposition~\ref{thm:vertical} the products $\wt
J_-\wt J_+$ and $\wt J_+\wt J_-$ have eigenvalues $(n+1)(n+2\ell+2)$ and $n(n+2\ell+1)$,
which are independent of any normalization of the states.

\begin{theorem}
\label{thm:su11-vertical}
Put $\lvert n\rangle_{\ell}:=\bigl(M^{(\ell)}_{n}\bigr)^{-1}f_{K\ell}$ with
$\bigl(M^{(\ell)}_{n}\bigr)^{2}=\binom{n+2\ell+1}{n}$, and endow their span with the
inner product which makes $\{\lvert n\rangle_{\ell}\}$ orthonormal (this is an abstract
inner product, unitarily equivalent to $\ell^{2}(\N_0)$, and not the $L^2(0,\pi)$ one).
Define
$J^{(\ell)}_{\pm}\lvert n\rangle_{\ell}:=\wt{J}^{\,K}_{\pm}\lvert n\rangle_{\ell}$ and
$J^{(\ell)}_{0}\lvert n\rangle_{\ell}:=(n+\ell+1)\lvert n\rangle_{\ell}$. Then
\[
J^{(\ell)}_{+}\lvert n\rangle_{\ell}=\sqrt{(n+1)(n+2\ell+2)}\,\lvert n+1\rangle_{\ell},
\quad
J^{(\ell)}_{-}\lvert n\rangle_{\ell}=\sqrt{n(n+2\ell+1)}\,\lvert n-1\rangle_{\ell},
\]
these operators verify $[J^{(\ell)}_{0},J^{(\ell)}_{\pm}]=\pm J^{(\ell)}_{\pm}$ and
$[J^{(\ell)}_{-},J^{(\ell)}_{+}]=2J^{(\ell)}_{0}$, and the Casimir
$Q_{\ell}=(J^{(\ell)}_{0})^{2}-\ha(J^{(\ell)}_{+}J^{(\ell)}_{-}+J^{(\ell)}_{-}J^{(\ell)}_{+})$
acts as $\ell(\ell+1)$. Thus, each column is, on the algebraic span of the $\lvert
n\rangle_{\ell}$, an irreducible lowest weight $\mathrm{su}(1,1)$ module with lowest
weight vector $\lvert0\rangle_{\ell}=\sin^{\ell+1}\!\chi$ and Bargmann index
$k=\ell+1$; that is, the underlying lowest weight module of the positive discrete series
representation $D^{+}_{\ell+1}$ in the classification of \cite{Bargmann} (see also
\cite[Ch.~6]{Vilenkin}).
\end{theorem}
\begin{proof}
From \eqref{eq:vertical} and
$\bigl(M^{(\ell)}_{n+1}/M^{(\ell)}_{n}\bigr)^{2}=(n+2\ell+2)/(n+1)$ one gets
$J^{(\ell)}_{+}\lvert n\rangle_{\ell}
=(n+1)\sqrt{\tfrac{n+2\ell+2}{n+1}}\,\lvert n+1\rangle_{\ell}$, and in the same way for
$J^{(\ell)}_{-}$, the two matrix elements being mutually adjoint. The weight relations
are immediate from the $\pm1$ shift, and
$[J^{(\ell)}_{-},J^{(\ell)}_{+}]\lvert n\rangle_{\ell}
=[(n+1)(n+2\ell+2)-n(n+2\ell+1)]\lvert n\rangle_{\ell}=2(n+\ell+1)\lvert n\rangle_{\ell}$.
As for the Casimir, $(n+\ell+1)^{2}-\ha[n(n+2\ell+1)+(n+1)(n+2\ell+2)]=\ell(\ell+1)$.
The $f_{K\ell}$ of one column are linearly independent, because by \eqref{eq:fKl}
they are $\sin^{\ell+1}\!\chi$ times polynomials in $\cos\chi$ of pairwise distinct
degrees $n$.
Let us consider now the irreducibility on the algebraic span, that is, among the
subspaces of finitely supported combinations of the $\lvert n\rangle_{\ell}$ (nothing is
claimed about closed invariant subspaces of the completion, where the generators are
unbounded). If $V\neq\{0\}$ is invariant, a Lagrange interpolation polynomial in
$J^{(\ell)}_{0}$ (whose eigenvalues $n+\ell+1$ are pairwise distinct) isolates some
$\lvert n\rangle_{\ell}\in V$. Now, the lowering coefficient vanishes only at $n=0$ and
the raising coefficient never vanishes, so a repeated application produces the whole
column. Finally, since
$\ha(J_{+}J_{-}+J_{-}J_{+})=J_1^{2}+J_2^{2}$ for $J_{\pm}=J_1\pm\ii J_2$, our Casimir is
$Q_\ell=J_0^{2}-J_1^{2}-J_2^{2}$, whose eigenvalue on the module of Bargmann index $k$ is
$k(k-1)$; hence $\ell(\ell+1)=k(k-1)$ and $k=\ell+1$. (The opposite sign convention for
the Casimir, also in use, would give $k(1-k)$ instead.)
\end{proof}

\subsection{The circular row}
\label{subsec:circular}

Along the row $n=0$ the situation is different in two respects. First, the differential
expression $\wt O^{\,-}_{K}$ does not annihilate the bottom state. Indeed, at $K=0$ one
has $\wt O^{\,-}_{0}=-\cos\chi\frac{\dd}{\dd\chi}-\sin\chi$ and
\begin{equation}
\wt O^{\,-}_{0}f_{00}=-\cos^{2}\chi-\sin^{2}\chi=-1,
\label{eq:O0-escape}
\end{equation}
which is the case $n=0$ of the general identity \eqref{eq:H-leaves-space} below, and it
lies outside the dressed space, because every $f_{K\ell}$ vanishes at $\chi=0$. Thus,
the lowest weight condition must be imposed in the abstract module, instead of being
inherited from the differential representative. Second, the raw operators
\eqref{eq:Odiag} do not close. A direct computation gives $[\wt O^{\,-},\wt
O^{\,+}]f_{KK}=(-8K-6)f_{KK}$, which is affine in $K$, but it has neither the slope nor
the sign of a diagonal $\mathrm{su}(1,1)$ generator. So, both a normalization of the
states and a rescaling of the operators are needed here.

\begin{proposition}
\label{thm:su11}
Put $\lvert K\rangle:=(K+1)^{-1/2}f_{KK}=(K+1)^{-1/2}\sin^{K+1}\!\chi$, endow the span
with the inner product which makes these orthonormal, and set
\[
\wt{C}^{\,+}_{K}:=\ha\,\wt{O}^{\,+}_{K},
\qquad
\wt{C}^{\,-}_{K}:=-\frac{K+1}{2K+1}\,\wt{O}^{\,-}_{K}\quad(K\ge1),
\qquad
C_{0}\lvert K\rangle:=(K+1)\lvert K\rangle ,
\]
completing the definition of the global lowering operator at the boundary by
$C^{\,-}\lvert0\rangle:=0$, which by \eqref{eq:O0-escape} is \emph{not} the action of
$\wt O^{\,-}_{0}$. The differential representative is used only for $K\ge1$. Then
$C^{\,+}\lvert K\rangle=\sqrt{(K+1)(K+2)}\lvert K+1\rangle$ and $C^{\,-}\lvert
K\rangle=\sqrt{K(K+1)}\lvert K-1\rangle$ for all $K\ge0$, the relations
$[C_{0},C^{\,\pm}]=\pm C^{\,\pm}$, $[C^{\,-},C^{\,+}]=2C_{0}$ hold, and the Casimir
vanishes. Thus, the circular states carry the irreducible lowest weight
$\mathrm{su}(1,1)$ module of Bargmann index $k=1$, with lowest weight vector
$\lvert0\rangle$, that is, the one underlying the positive discrete series representation
$D^{+}_{1}$ of \cite{Bargmann}.
\end{proposition}
\begin{proof}
The proof is basically the same as that of Theorem~\ref{thm:su11-vertical}, starting from
\eqref{eq:Odiag} and the rescalings we have just stated, the case $K=0$ of the lowering
action being the boundary prescription and not a computation. The states
$f_{KK}=\sin^{K+1}\!\chi$ are linearly independent.
\end{proof}

\begin{remark}\label{rem:matrix-elements}
The circular realization is not canonical. Once $C_0|K\rangle=(K+1)|K\rangle$ has been
fixed, the adjointness and the commutator determine the matrix elements in
Proposition~\ref{thm:su11}. If instead we choose $C_0|K\rangle=(K+k)|K\rangle$ with
$k>0$, the same argument produces the lowest weight module of Bargmann index $k$ after a
suitable label dependent rescaling; this module underlies a representation of $SU(1,1)$
itself when $2k\in\N$, and one of its universal covering group otherwise
\cite{Bargmann}. Thus, only the raising and lowering pattern, including the imposed
boundary value $C^-|0\rangle=0$, is intrinsic to this construction. On the other hand, the
column index $k=\ell+1$ of Theorem~\ref{thm:su11-vertical} is a positive integer, and it
is fixed without any such rescaling.
\end{remark}

\subsection{Rows}
The horizontal descent does not close on the dressed space. At $\ell=0$ one has
$f_{n0}=\sin((n+1)\chi)$, and a direct computation gives
\begin{equation}\label{eq:H-leaves-space}
\wt H_-^{\,n0}f_{n0}=-(n+1)T_n(\cos\chi).
\end{equation}
This function does not vanish at $\chi=0$, while every $f_{K\ell}$ does.
\noindent
For $n=0$, \eqref{eq:H-leaves-space} is \eqref{eq:O0-escape}. A row with $n\ge1$ presents
the two features already met in Section~\ref{subsec:circular} for the circular row.
First, the differential representative does not annihilate the bottom state. Second, the
raw operators do not close: composing \eqref{eq:horiz} at consecutive labels gives
\begin{equation}
[\wt H^{-},\wt H^{+}]f_{K\ell}=-(8\ell+4n+6)f_{K\ell},
\label{eq:horiz-commutator}
\end{equation}
which is affine in $\ell$, but with neither the slope nor the sign of a diagonal
$\mathrm{su}(1,1)$ generator; at $n=0$ and $\ell=K$ it is the value $-(8K+6)$ of
Section~\ref{subsec:circular}. Thus, a normalization of the states and a rescaling of the
operators are needed, and it turns out that the ones used in Proposition~\ref{thm:su11} apply along any row.

\begin{proposition}
\label{prop:rows}
Fix $n\ge0$ and $k>0$, and put $v_\ell:=f_{(n+\ell)\ell}$,
$\wt H^{(\ell)}_{\pm}:=\wt H^{\,(n+\ell)\ell}_{\pm}$ for $\ell\ge0$. Put
$\lvert\ell\rangle_{n}:=\bigl((2k)_\ell/\ell!\bigr)^{-1/2}v_\ell$, endow their span with
the inner product which makes these orthonormal, and set
\[
\begin{split}
&\wt B^{\,+}_{\ell}:=\ha\,\wt H^{(\ell)}_{+},
\qquad
B_{0}\lvert\ell\rangle_{n}:=(\ell+k)\lvert\ell\rangle_{n},
\\
&\wt B^{\,-}_{\ell}:=-\frac{2\ell(\ell+2k-1)}{(2\ell+n+1)(2\ell+n)}\,\wt H^{(\ell)}_{-}
\qquad(\ell\ge1),
\end{split}
\]
completing the definition of the global lowering operator at the boundary by
$B^{-}\lvert0\rangle_{n}:=0$, which by \eqref{eq:H-leaves-space} is \emph{not} the action
of $\wt H^{\,n0}_{-}$. The differential representative is used only for $\ell\ge1$. Then
\[
B^{+}\lvert\ell\rangle_{n}=\sqrt{(\ell+1)(\ell+2k)}\,\lvert\ell+1\rangle_{n},
\qquad
B^{-}\lvert\ell\rangle_{n}=\sqrt{\ell(\ell+2k-1)}\,\lvert\ell-1\rangle_{n}
\]
for all $\ell\ge0$, the relations $[B_{0},B^{\pm}]=\pm B^{\pm}$,
$[B^{-},B^{+}]=2B_{0}$ hold, and the Casimir acts as $k(k-1)$. Thus each row carries an
irreducible lowest weight $\mathrm{su}(1,1)$ module of Bargmann index $k$, with lowest
weight vector $\lvert0\rangle_{n}$. For $n=0$ and $k=1$ the rescalings are those of
Proposition~\ref{thm:su11}, which is recovered.
\end{proposition}
\begin{proof}
Along a row $K=n+\ell$, so \eqref{eq:horiz} reads
$\wt H^{(\ell)}_{+}v_\ell=2(\ell+1)v_{\ell+1}$ and
$\wt H^{(\ell)}_{-}v_\ell=-\frac{(2\ell+n+1)(2\ell+n)}{2\ell}v_{\ell-1}$ for $\ell\ge1$,
and the rescalings give $B^{+}v_\ell=(\ell+1)v_{\ell+1}$,
$B^{-}v_\ell=(\ell+2k-1)v_{\ell-1}$. The $v_\ell$ are linearly independent: by
\eqref{eq:fKl} and $\Geg{n}{\ell+1}(1)=(2\ell+2)_n/n!=\binom{n+2\ell+1}{n}$,
$v_\ell=\binom{n+2\ell+1}{n}\chi^{\ell+1}+O(\chi^{\ell+2})$ as $\chi\to0^{+}$, so their
orders of vanishing are pairwise distinct. From
$\bigl((2k)_{\ell+1}/(\ell+1)!\bigr)\bigl((2k)_{\ell}/\ell!\bigr)^{-1}
=(\ell+2k)/(\ell+1)$ one obtains the two displayed matrix elements, which are mutually
adjoint. The weight relations follow from the $\pm1$ shift. For $\ell\ge1$,
$[B^{-},B^{+}]\lvert\ell\rangle_{n}
=\bigl[(\ell+1)(\ell+2k)-\ell(\ell+2k-1)\bigr]\lvert\ell\rangle_{n}
=2(\ell+k)\lvert\ell\rangle_{n}$; at $\ell=0$ the prescription gives
$B^{-}B^{+}\lvert0\rangle_{n}=2k\lvert0\rangle_{n}$ and
$B^{+}B^{-}\lvert0\rangle_{n}=0$. For the Casimir,
$(\ell+k)^{2}-\ha\bigl[\ell(\ell+2k-1)+(\ell+1)(\ell+2k)\bigr]=k(k-1)$. Irreducibility on
the algebraic span follows as in Theorem~\ref{thm:su11-vertical}, the lowering
coefficient vanishing only at $\ell=0$ and the raising coefficient never.
\end{proof}

\begin{remark}
\label{rem:intrinsic-vs-prescribed}
According to these results, along a column the lowering
representative annihilates the bottom state, $\wt J^{\,\ell}_{-}f_{\ell\ell}=0$ by
Proposition~\ref{thm:vertical}, no rescaling is needed, and the Bargmann index takes the
integer value $k=\ell+1$ (Theorem~\ref{thm:su11-vertical}). Along a row,
the bottom state is not annihilated by the differential representative, the lowest weight
condition is imposed in the abstract module, and $k>0$ is free, as in
Remark~\ref{rem:matrix-elements}. For $\ell\ge1$ the differential representatives are of
first order in both cases. On a row with $n\ge1$ the choice $k=n$ gives Casimir $n(n-1)$
and the matrix elements $\sqrt{(\ell+1)(2n+\ell)}$ and $\sqrt{\ell(2n+\ell-1)}$. Finally,
this subsection concerns the undeformed functions $f_{K\ell}$; the deformed family
\eqref{eq:inlevel-F} with $c\neq0$ is treated next.
\end{remark}

\subsection[The deformed lattice]{The deformed lattice}
\label{subsec:deformed}
\noindent
For the deformed functions \eqref{eq:inlevel-F} the first order path of
Proposition~\ref{prop:rows} is not available when $n\ge1$, as the conjugated horizontal
operator does not preserve the row. For $n=1$, $\ell=0$ a direct computation gives
\begin{equation}
e^{-\alpha_{2}\chi/2}\,\wt H^{\,10}_{+}\,e^{\alpha_{1}\chi/2}F_{10}
=e^{-\alpha_{2}\chi/2}\sin^{3}\!\chi
\Bigl[\tfrac{c}{4}\cot^{2}\chi+4\cot\chi-\tfrac{3c}{4}\Bigr],
\label{eq:first-order-fails}
\end{equation}
of degree $2$ in $\cot\chi$ for $c\neq0$, whereas $F_{21}$ has degree $n=1$; at $c=0$ the
two outer terms inside the square bracket vanish and \eqref{eq:first-order-fails} is a multiple of $f_{21}$.

A row is nevertheless the carrier of the same module, and can be traversed using a composition of operators already defined. By \eqref{eq:inlevel-A-act} the operator $\wt A^{\,-}_{\ell}$ preserves the level
and lowers $n$ by one, and by \eqref{eq:inlevel-A-act2} the operator $\wt A^{\,+}_{\ell}$
preserves the level and raises $n$ by one; and $F_{KK}=e^{-\alpha_K\chi/2}\sin^{K+1}\!\chi$
differs from $f_{KK}$ by the level dressing alone, so that \eqref{eq:Odiag} moves along
the circular row after conjugation by that dressing. Descending to the circular row,
moving one step along it and climbing back shifts $\ell$ by one at fixed $n$. Accordingly,
for $n\ge0$ and $K=n+\ell$ we set
\begin{equation}
\begin{split}
\wt{\mathcal C}^{\,+}_{K\ell}
&:=\wt A^{\,+}_{\ell+2}\cdots\wt A^{\,+}_{K+1}\circ
e^{-\alpha_{K+1}\chi/2}\,\wt O^{\,+}_{K}\,e^{\alpha_{K}\chi/2}\circ
\wt A^{\,-}_{K-1}\cdots\wt A^{\,-}_{\ell}\,,
\\
\wt{\mathcal C}^{\,-}_{K\ell}
&:=\wt A^{\,+}_{\ell}\cdots\wt A^{\,+}_{K-1}\circ
e^{-\alpha_{K-1}\chi/2}\,\wt O^{\,-}_{K}\,e^{\alpha_{K}\chi/2}\circ
\wt A^{\,-}_{K-1}\cdots\wt A^{\,-}_{\ell}\qquad(\ell\ge1),
\end{split}
\label{eq:deformed-composite}
\end{equation}
with the convention that the rightmost factor acts first and that an empty product is the
identity. Each is a differential operator of order $2n+1$.

\begin{proposition}
\label{prop:deformed-rows}
Put
\[
\Pi_{K\ell}:=n!\,(n+2\ell+2)_n\prod_{j=\ell+1}^{K}
\Bigl[1+\frac{\alpha_K^{2}}{4j^{2}}\Bigr] .
\]
Then, for $n\ge0$ and $K=n+\ell$,
\begin{equation}
\begin{split}
&\wt{\mathcal C}^{\,+}_{K\ell}F_{K\ell}
=\frac{2(K+1)(2K+3)}{2\ell+3}\,\Pi_{K\ell}\,F_{(K+1)(\ell+1)},
\\
&\wt{\mathcal C}^{\,-}_{K\ell}F_{K\ell}
=-(2\ell+1)\,\Pi_{K\ell}\,F_{(K-1)(\ell-1)}\qquad(\ell\ge1),
\end{split}
\label{eq:deformed-actions}
\end{equation}
and both constants are nonzero.
\end{proposition}
\begin{proof}
Along the descent the level $K$ is fixed and $\ell'$ runs from $\ell$ to $K-1$, the degree
being $n'=K-\ell'$; by \eqref{eq:inlevel-A-act} the step at $\ell'$ contributes
$-\frac{(K-\ell')(K+\ell'+2)}{2\ell'+3}\bigl[1+\alpha_K^{2}/(4(\ell'+1)^{2})\bigr]$, and
the descent ends at $F_{KK}$. Multiplication by $e^{\alpha_K\chi/2}$ gives
$\sin^{K+1}\!\chi=f_{KK}$, and \eqref{eq:Odiag} contributes $2(K+1)$ in the first case and
$-(2K+1)$ in the second; multiplication by $e^{-\alpha_{K\pm1}\chi/2}$ produces
$F_{(K\pm1)(K\pm1)}$. Along the climb the level is $K\pm1$ and \eqref{eq:inlevel-A-act2}
contributes $-(2\ell'+1)$ at each step, with $\ell'$ running from $K+1$ down to $\ell+2$
in the first case and from $K-1$ down to $\ell$ in the second. Collecting,
$\prod_{\ell'=\ell}^{K-1}(K-\ell')=n!$ and
$\prod_{\ell'=\ell}^{K-1}(K+\ell'+2)=(n+2\ell+2)_n$, while the odd factors telescope,
\[
\frac{\prod_{\ell'=\ell+2}^{K+1}(2\ell'+1)}{\prod_{\ell'=\ell}^{K-1}(2\ell'+3)}
=\frac{2K+3}{2\ell+3},
\qquad
\frac{\prod_{\ell'=\ell}^{K-1}(2\ell'+1)}{\prod_{\ell'=\ell}^{K-1}(2\ell'+3)}
=\frac{2\ell+1}{2K+1},
\]
and the $2n$ signs $(-1)^{n}(-1)^{n}$ cancel. This gives \eqref{eq:deformed-actions}. Every
factor of $\Pi_{K\ell}$ is positive, so neither constant vanishes.
\end{proof}

\begin{proposition}
\label{prop:deformed-module}
Fix $n\ge0$ and $k>0$, and put $w_\ell:=F_{(n+\ell)\ell}$. Rescaling
\eqref{eq:deformed-composite} as in Proposition~\ref{prop:rows}, that is, setting
\[
\mathcal B^{+}w_\ell:=\frac{(\ell+1)(2\ell+3)}{2(K+1)(2K+3)\,\Pi_{K\ell}}
\wt{\mathcal C}^{\,+}_{K\ell}w_\ell,
\qquad
\mathcal B^{-}w_\ell:=-\frac{\ell+2k-1}{(2\ell+1)\,\Pi_{K\ell}}
\wt{\mathcal C}^{\,-}_{K\ell}w_\ell\ \ (\ell\ge1),
\]
with $\mathcal B^{-}w_0:=0$ and $\mathcal B_{0}w_\ell:=(\ell+k)w_\ell$, one has
$\mathcal B^{+}w_\ell=(\ell+1)w_{\ell+1}$ and
$\mathcal B^{-}w_\ell=(\ell+2k-1)w_{\ell-1}$, and the conclusions of
Proposition~\ref{prop:rows} hold verbatim: in the basis
$\bigl((2k)_\ell/\ell!\bigr)^{-1/2}w_\ell$ the matrix elements are
$\sqrt{(\ell+1)(\ell+2k)}$ and $\sqrt{\ell(\ell+2k-1)}$, the commutation relations of
$\mathrm{su}(1,1)$ hold, and the Casimir acts as $k(k-1)$.
\end{proposition}
\begin{proof}
The two actions follow from \eqref{eq:deformed-actions}. The $w_\ell$ are linearly
independent: $\Rom{n}{\alpha_K}{-K}$ is monic, so $\Rom{n}{\alpha_K}{-K}(\cot\chi)$ is
$\chi^{-n}+O(\chi^{-n+1})$ as $\chi\to0^{+}$ and $w_\ell=\chi^{\ell+1}+O(\chi^{\ell+2})$,
whence their orders of vanishing are pairwise distinct. The remaining computations are
those of Proposition~\ref{prop:rows}, which involve only the two displayed actions.
\end{proof}

\begin{remark}
\label{rem:deformed}
At $n=0$ one has $\Pi_{K\ell}=1$ and the constants of \eqref{eq:deformed-actions} reduce
to $2(K+1)$ and $-(2K+1)$, the constants of \eqref{eq:Odiag}, while
\eqref{eq:deformed-composite} reduces to the conjugation of $\wt O^{\,\pm}_{K}$ by the
level dressing, an operator of first order. For $n\ge1$ the order is $2n+1$; we do not
claim it to be minimal. Taking $k=n$ for $n\ge1$ gives Casimir $n(n-1)$ and the matrix
elements $\sqrt{(\ell+1)(2n+\ell)}$ and $\sqrt{\ell(2n+\ell-1)}$.
\end{remark}

The columns of the deformed lattice are reached in the same way, the only change being
the length of the climb. For $n\ge0$ and $K=n+\ell$ we set
\begin{equation}
\begin{split}
\wt{\mathcal B}^{\,+}_{K\ell}
&:=\wt A^{\,+}_{\ell+1}\cdots\wt A^{\,+}_{K+1}\circ
e^{-\alpha_{K+1}\chi/2}\,\wt O^{\,+}_{K}\,e^{\alpha_{K}\chi/2}\circ
\wt A^{\,-}_{K-1}\cdots\wt A^{\,-}_{\ell}\,,
\\
\wt{\mathcal B}^{\,-}_{K\ell}
&:=\wt A^{\,+}_{\ell+1}\cdots\wt A^{\,+}_{K-1}\circ
e^{-\alpha_{K-1}\chi/2}\,\wt O^{\,-}_{K}\,e^{\alpha_{K}\chi/2}\circ
\wt A^{\,-}_{K-1}\cdots\wt A^{\,-}_{\ell}\qquad(n\ge1),
\end{split}
\label{eq:deformed-column-composite}
\end{equation}
with the same conventions as in \eqref{eq:deformed-composite}. These are differential
operators of order $2n+2$ and $2n$ respectively.

\begin{proposition}
\label{prop:deformed-columns}
With $\Pi_{K\ell}$ as in Proposition~\ref{prop:deformed-rows},
\begin{equation}
\wt{\mathcal B}^{\,+}_{K\ell}F_{K\ell}=-2(K+1)(2K+3)\,\Pi_{K\ell}\,F_{(K+1)\ell},
\qquad
\wt{\mathcal B}^{\,-}_{K\ell}F_{K\ell}=\Pi_{K\ell}\,F_{(K-1)\ell}\quad(n\ge1),
\label{eq:deformed-column-actions}
\end{equation}
and both constants are nonzero. After the normalization of
Theorem~\ref{thm:su11-vertical}, each column of the deformed lattice carries the
irreducible lowest weight $\mathrm{su}(1,1)$ module of Bargmann index $k=\ell+1$, with
matrix elements $\sqrt{(n+1)(n+2\ell+2)}$ and $\sqrt{n(n+2\ell+1)}$ and lowest weight
vector $F_{\ell\ell}$.
\end{proposition}
\begin{proof}
The descent is that of Proposition~\ref{prop:deformed-rows} and contributes the same
factor; the circular step contributes $2(K+1)$ and $-(2K+1)$ as before. The climb now runs
over $\ell'=K+1,\dots,\ell+1$ in the first case and $\ell'=K-1,\dots,\ell+1$ in the
second, so that the odd factors telescope to
\[
\frac{\prod_{\ell'=\ell+1}^{K+1}(2\ell'+1)}{\prod_{\ell'=\ell}^{K-1}(2\ell'+3)}=2K+3,
\qquad
\frac{\prod_{\ell'=\ell+1}^{K-1}(2\ell'+1)}{\prod_{\ell'=\ell}^{K-1}(2\ell'+3)}
=\frac{1}{2K+1},
\]
and the $2n+1$ and $2n-1$ signs leave an overall $-1$ in the first case and, combined with
$-(2K+1)$, a $+1$ in the second. This gives \eqref{eq:deformed-column-actions}. For the
module, the $F_{K\ell}$ of one column are linearly independent because by
Remark~\ref{rem:factorization} they are eigenfunctions of the single operator $H_\ell$
with eigenvalues $E_K$, and $E_K$ is strictly increasing in $K$, its derivative with
respect to $(K+1)^{2}$ being $1+c^{2}/\bigl(4(K+1)^{4}\bigr)>0$. The remaining
computations are those of Theorem~\ref{thm:su11-vertical}, which involve only the two
displayed actions; in particular the index $k=\ell+1$ is again forced.
\end{proof}

\begin{corollary}
\label{cor:recurrences}
Conjugating by the level dressing, that is, putting
\[
\mathcal R^{\,+}_{K\ell}:=\sin^{-(K+2)}\!\chi\;e^{\alpha_{K+1}\chi/2}\circ
\wt{\mathcal B}^{\,+}_{K\ell}\circ e^{-\alpha_{K}\chi/2}\sin^{K+1}\!\chi ,
\]
and similarly for the other three composites, one obtains relations between Romanovski
polynomials whose parameter pairs differ. For instance
\begin{equation}
\mathcal R^{\,+}_{K\ell}\,\Rom{n}{\alpha_K}{-K}(\cot\chi)
=-2(K+1)(2K+3)\,\Pi_{K\ell}\;\Rom{n+1}{\alpha_{K+1}}{-(K+1)}(\cot\chi),
\label{eq:romanovski-recurrence}
\end{equation}
and the smallest instance, at $K=\ell=n=0$, reads
$\mathcal R^{\,+}_{00}\mathbf 1=-6\bigl(\cot\chi-\tfrac{c}{4}\bigr)
=-6\,\Rom{1}{\alpha_1}{-1}(\cot\chi)$.
\end{corollary}

\begin{remark}
\label{rem:recurrence-nature}
The two sides in \eqref{eq:romanovski-recurrence} carry the \emph{different} parameter pairs $(\alpha_K,-K)$ and
$(\alpha_{K+1},-(K+1))$, so it is a relation different from the classical fixed
parameter relations, as recalled in the Introduction. It is also worth noticing that the conjugation leaves the
net factor $e^{(\alpha_{K+1}-\alpha_{K})\chi/2}$, that is,
$\exp\bigl[\tfrac12(\alpha_{K+1}-\alpha_K)\arccot x\bigr]$ in the variable $x=\cot\chi$,
which is a power of the exponential factor of the weight \eqref{eq:weight}. The
coefficients of $\mathcal R^{\,\pm}_{K\ell}$ are therefore not rational in $x$, and the
relation is stated in the variable $\chi$.
\end{remark}

Collecting Theorem~\ref{thm:su11-vertical}, Proposition~\ref{prop:rows},
Proposition~\ref{prop:deformed-module} and Proposition~\ref{prop:deformed-columns}, the
four one-parameter families of the lattice carry lowest weight $\mathrm{su}(1,1)$
modules as follows, the entries giving the order of the differential representative and
the Bargmann index.
\begin{center}
\small
\begin{tabular}{@{}lll@{}}
\hline
& undeformed $f_{K\ell}$ & deformed $F_{K\ell}$\\
\hline
column, $\ell$ fixed & order $1$, index $\ell+1$
                     & orders $2n+2$, $2n$, index $\ell+1$\\
row, $n$ fixed       & order $1$, index arbitrary
                     & order $2n+1$, index arbitrary\\
\hline
\end{tabular}
\end{center}
Along a column the Bargmann index is forced to $\ell+1$; along a row it may be any
positive number. The deformation raises the order of the differential representative, but it changes neither the module nor the status of the index.

\appendix 
\section{Change of variable and conjugation}\label{app:conjugation} 

Here we describe explicitly how the Gegenbauer operators in 
Section~\ref{sec:ladders} are transformed into operators acting on the dressed functions. 
Under the change of variable $u=\cos\chi$ one has 
\[ 
\frac{\dd}{\dd u} = -\frac{1}{\sin\chi}\frac{\dd}{\dd\chi}, 
\qquad 
(1-u^2)\frac{\dd}{\dd u} = -\sin\chi\frac{\dd}{\dd\chi}, 
\qquad 
u\frac{\dd}{\dd u} = -\cot\chi\frac{\dd}{\dd\chi}. 
\] 
For any real number $s$, let $M_s$ denote multiplication by $\sin^s\chi$: $(M_sg)(\chi):=\sin^s\chi\,g(\chi)$. If $L\Geg{n}{\lambda}=a_{n,\lambda}\Geg{m}{\mu}$ is a Gegenbauer relation, then the corresponding dressed functions are $f=M_\lambda\Geg{n}{\lambda}(\cos\chi)$ and $\widehat f=M_\mu\Geg{m}{\mu}(\cos\chi)$. The operator acting between the dressed functions is therefore 
\begin{equation} 
\widetilde L := M_\mu L M_\lambda^{-1}, 
\label{app:eq:transformed-operator} 
\end{equation} 
because $\widetilde L f = M_\mu L M_\lambda^{-1} \bigl(M_\lambda\Geg{n}{\lambda}\bigr) = M_\mu L\Geg{n}{\lambda} = a_{n,\lambda}\widehat f$.
Thus, it is the composition $M_\mu L M_\lambda^{-1}$ that transforms the original 
Gegenbauer operator into the operator acting on the dressed functions. When $\mu=\lambda$, 
formula \eqref{app:eq:transformed-operator} is an ordinary conjugation. In that case, for 
a first order operator $L=A(\chi)\dd/\dd\chi +B(\chi)$ one obtains 
\begin{equation} 
M_sLM_s^{-1} = A(\chi)\frac{\dd}{\dd\chi} +B(\chi)-sA(\chi)\cot\chi. 
\label{app:eq:conj-gen} 
\end{equation} 
Indeed, 
\[ 
M_s\frac{\dd}{\dd\chi}M_s^{-1} = \frac{\dd}{\dd\chi}-s\cot\chi. 
\] 
For the vertical relations, the Gegenbauer order $\lambda=\ell+1$ is unchanged. The source 
and target functions therefore have the same dressing $M_{\ell+1}$, and the transformed 
operators are obtained by the conjugation 
$\widetilde J_{\pm}^{\,K} = M_{\ell+1}J_{\pm}M_{\ell+1}^{-1}$. 
Applying \eqref{app:eq:conj-gen} to the two fixed-order Gegenbauer relations displayed in 
the proof of Proposition~\ref{thm:vertical} gives 
$\widetilde J_{\pm}^{\,K} = \pm\sin\chi\dd/\dd\chi +(K+1)\cos\chi$. The common coefficient 
$K+1=n+\ell+1$ results from combining the zeroth order term of the original Gegenbauer 
operator with the term introduced by the conjugation. For the horizontal relations, the 
degree $n$ is unchanged but the Gegenbauer order changes. If the order is raised from 
$\lambda$ to $\lambda+1$, the source is dressed by $M_\lambda$ and the target by 
$M_{\lambda+1}$. Hence the operator acting on the dressed functions is 
$\widetilde H_+^{\,K\ell} = M_{\lambda+1}H_+M_\lambda^{-1}$, for $\lambda=\ell+1$. 
Likewise, for the 
order-lowering relation, $\widetilde H_-^{\,K\ell} = M_{\lambda-1}H_-M_\lambda^{-1}$. Note 
that these are not ordinary conjugations because the multiplication operators on the left 
and on the right are different. Substituting $u=\cos\chi$ into the two Gegenbauer order-
shifting relations and performing the compositions above gives 
\[ 
\widetilde H_+^{\,K\ell} = -\cos\chi\frac{\dd}{\dd\chi} +\frac{\ell+1}{\sin\chi} +(K+1)\sin\chi 
\] 
and 
\[ 
\widetilde H_-^{\,K\ell} = -\cos\chi\frac{\dd}{\dd\chi} -\frac{\ell}{\sin\chi} -(K+1)\sin\chi. 
\] 
The in-level operators are obtained in the same way. The derivative relation 
$\frac{\dd}{\dd u}\Geg{n}{\lambda} = 2\lambda\Geg{n-1}{\lambda+1}$ raises the order 
from $\lambda$ to $\lambda+1$. Therefore its transformation is 
\[ 
M_{\lambda+1} \left(-\frac{1}{\sin\chi}\frac{\dd}{\dd\chi}\right) M_\lambda^{-1} = -\frac{\dd}{\dd\chi}+\lambda\cot\chi. 
\] 
With $\lambda=\ell+1$, this is precisely $\widetilde\Lambda_\ell^+ = -\dd/\dd\chi+(\ell+1)\cot\chi$. Similarly, transforming the weighted companion relation, which lowers the order from $\lambda$ to $\lambda-1$, gives $\widetilde\Lambda_\ell^- = \dd/\dd\chi +\ell\cot\chi$. These calculations also explain why the transformed operators depend on $n$ and $\ell$ only through the labels $(K,\ell)$, with $K=n+\ell$.

\subsection*{Acknowledgements}
The authors express their gratitude to Prof. Kenier Castillo for his careful revision 
of the manuscript as well as numerous comments that improved the original results; in 
particular, he suggested the contents of subsection \ref{subsec:jacobi} and proposition
\ref{prop:rows}, as well as the need for subsection \ref{subsec:deformed}.

\subsection*{Disclosure statement}
The authors report there are no competing interests to declare. This
research received no specific grant from funding agencies in the public,
commercial, or not-for-profit sectors. During the preparation of this
manuscript, the authors used Claude Opus 5 (Anthropic) to check for typos,
grammar errors, and LaTeX syntax and cross-references. 
All mathematical content, derivations,
and scientific conclusions are the authors’ own work, and the authors take
full responsibility for the final text.

{\small\raggedright

}

\end{document}